\documentclass[runningheads,orivec,envcountsame]{llncs}
\usepackage[T1]{fontenc}
\usepackage{graphicx}
\usepackage{hyperref}
\usepackage{color}

\usepackage{amsthm}

\usepackage{amssymb}
\usepackage{xspace}
\usepackage{relsize}
\usepackage{scalefnt}
\usepackage{todonotes}
\usepackage{mathtools}
\usepackage{booktabs}
\usepackage{comment}
\usepackage{enumerate}
\newcommand{\logicFont}[1]{\protect\ensuremath{\mathrm{#1}}\xspace}
\newcommand{\problemmaFont}[1]{\protect\ensuremath{\mathsf{#1}}\xspace}

\newcommand{\classFont}[1]{\protect\ensuremath{\mathsf{#1}}\xspace}

\newcommand{\existso}{\ensuremath{\exists^1}\xspace}
\newcommand{\forallo}{\ensuremath{\forall^1}\xspace}
\newcommand{\wcn}{\ensuremath{\mathop{\dot\sim}}\xspace}
\newcommand{\bn}{\ensuremath{\mathop{\sim}}\xspace}
\newcommand\ScaleExists[1]{\vcenter{\hbox{\scalefont{#1}$\exists$}}}
\DeclareMathOperator*\bigexists{%
  \vphantom\sum
  \mathchoice{\ScaleExists{2}}{\ScaleExists{1.4}}{\ScaleExists{1}}{\ScaleExists{0.75}}}

\newcommand{\dblsetminus}{\mathbin{{\setminus}\mspace{-5mu}{\setminus}}}
\newcommand{\vvee}{\raisebox{1pt}{\ensuremath{\,\mathop{\mathsmaller{\mathsmaller{\dblsetminus\hspace{-0.23ex}/}}}\,}}}

\newcommand{\PL}{\logicFont{PL}}

\newcommand{\pind}{\perp\!\!\!\perp}
\newcommand{\cpind}{\perp\!\!\!\perp_{\mathrm{c}}}

\newcommand {\pci}[3] {#2~\!\!\pind_{#1}\!\!~#3}
\newcommand {\pmi}[2] {#1~\!\!\pind\!\!~#2}

\newcommand{\SUM}{\mathrm{SUM}}
\newcommand{\tuple}[1] {\protect\ensuremath{\mathbf{#1}}}
\newcommand{\x}{\tuple x}
\newcommand{\y}{\tuple y}
\newcommand{\z}{\tuple z}

\newcommand{\SAT}{\problemmaFont{SAT}}

\newcommand{\MC}{\problemmaFont{MC}}

\newcommand{\VAL}{\problemmaFont{VAL}}

\newcommand{\NEXPTIME}{\classFont{NEXPTIME}}
\newcommand{\EXPTIME}{\classFont{EXPTIME}}

\newcommand{\PSPACE}{\classFont{PSPACE}}

\newcommand{\EXPSPACE}{\classFont{EXPSPACE}}
\newcommand{\twoEXPSPACE}{\classFont{2\text-EXPSPACE}}
\newcommand{\threeEXPSPACE}{\classFont{3\text-EXPSPACE}}

\newcommand{\AEXPTIME}{\classFont{AEXPTIME}}

\newcommand{\RE}{\classFont{RE}}
\newcommand{\coRE}{\classFont{coRE}}
\newcommand{\bbR}{\mathbb{R}}

\newcommand{\var}{\operatorname{Var}}
\newcommand{\Fr}{\operatorname{Fr}}
\newcommand{\ar}{\operatorname{Ar}}

\newcommand{\A}{\mathcal{A}}

\newcommand{\X}{\mathbb{X}}

\newcommand{\Y}{\mathbb{Y}}

\def\dep{=\!\!}
\newcommand{\supp}{\mathrm{supp}}

\newcommand{\entropyteam}[2]{\eH_{#2}(#1)}
\newcommand{\eH}{\mathrm{H}}
\newcommand{\entropy}[1]{\eH(#1)}

\newcommand{\entropyatom}[2]{\entropy{#1} = \entropy{#2}}

\newcommand{\ddfn}{\Coloneqq}
\newcommand{\dfn}{\coloneqq}

\newcommand{\probinclogic}{\FO (\approx)}

\newcommand{\poly}{\mathrm{poly}}

\def\dep{=\!\!}
\newcommand{\FO}{{\mathrm{FO}}}
\newcommand{\FOPT}{{\mathrm{FOPT}(\leq_{c}^\delta)}}

\newcommand{\ESO}{{\mathrm{ESO}}}
\newcommand{\SO}{{\mathrm{SO}}}

\newcommand{\ESOr}{\ESO_{\mathbb{R}}(\SUM,+,\times)}

\newcommand{\f}{\phi}

\title{
Logics with probabilistic team semantics\\ and the Boolean negation
}
\author{Miika Hannula\inst{1}\orcidID{0000-0002-9637-6664} \and
Minna Hirvonen\inst{1}\orcidID{0000-0002-2701-9620} \and
Juha Kontinen\inst{1}\orcidID{0000-0003-0115-5154} \and
Yasir Mahmood\inst{2}\orcidID{0000-0002-5651-5391} \and
Arne Meier\inst{3}\orcidID{0000-0002-8061-5376}\and
Jonni Virtema\inst{4}\orcidID{0000-0002-1582-3718}}
\authorrunning{M. Hannula et al.}
\institute{University of Helsinki, Department of Mathematics and Statistics, FI \and
Paderborn University, DICE group, Department of Computer Science, DE\and
Leibniz Universit\"at Hannover, Institut f\"ur Theoretische Informatik, Hannover, DE\and
University of Sheffield, Department of Computer Science, Sheffield, UK\\
\email{\{miika.hannula,minna.hirvonen,juha.kontinen\}@helsinki.fi}, \email{yasir.mahmood@uni-paderborn.de}, \email{meier@thi.uni-hannover.de}, \email{j.t.virtema@sheffield.ac.uk}}

\begin{document}
\maketitle              
\begin{abstract}
We study the expressivity and the complexity of various logics in probabilistic team semantics with the Boolean negation. 
In particular, we study the extension of probabilistic independence logic with the Boolean negation, and a recently introduced logic FOPT. 
We give a comprehensive picture of the relative expressivity of these logics together with the most studied logics in probabilistic team semantics setting, as well as relating their expressivity to a numerical variant of second-order logic. 
In addition, we introduce novel entropy atoms and show that the extension of first-order logic by entropy atoms subsumes probabilistic independence logic. 
Finally, we obtain some results on the complexity of model checking, validity, and satisfiability of our logics.

\end{abstract}
\section{Introduction}

Probabilistic team semantics is a novel framework for the logical analysis of probabilistic and quantitative dependencies.
Team semantics, as a semantic framework for logics involving qualitative dependencies and independencies, was introduced by Hodges \cite{hodges97} and popularised by V\"a\"an\"anen \cite{vaananen07} via his dependence logic. 
Team semantics defines truth in reference to collections of assignments, called \emph{teams}, and is particularly suitable for the formal analysis of properties, such as the functional dependence between variables, that arise only in the presence of multiple assignments. 
The idea of generalising team semantics to the probabilistic setting can be traced back to the works of Galliani \cite{galliani08} and Hyttinen et al.~\cite{Hyttinen15b}, however the beginning of a more systematic study of the topic dates back to works of Durand~et~al.~\cite{DHKMV18}.

In \emph{probabilistic team semantics} the basic semantic units are probability distributions (i.e., \emph{probabilistic teams}). 
This shift from set-based to distribution-based semantics allows probabilistic notions of dependency, such as conditional probabilistic independence, to be embedded in the framework\footnote{In \cite{Li22a} Li recently introduced \emph{first-order theory of random variables with probabilistic independence (FOTPI)} whose variables are interpreted by discrete distributions over the unit interval. The paper shows that true arithmetic is interpretable in FOTPI whereas probabilistic independence logic is by our results far less complex.}. 
The expressivity and complexity of non-probabilistic team-based logics can be related to fragments of (existential) second-order logic and have been studied extensively (see, e.g., \cite{GallianiH13,DurandKRV22,HannulaH22}). 
Team-based logics, by definition, are usually not closed under Boolean negation, so adding it can greatly increase the complexity and expressivity of these logics \cite{KontinenN11,HannulaKVV18}. 
Some expressivity and complexity  results have also been obtained for logics in  probabilistic team semantics (see below). 
However, richer semantic and computational frameworks are sometimes needed to characterise these logics.

\emph{Metafinite Model Theory}, introduced by Gr\"adel and Gurevich \cite{GradelG98}, generalises the approach of \emph{Finite Model Theory} by shifting to two-sorted structures, which extend finite structures by another (often infinite) numerical domain 
and weight functions bridging the two sorts. 
A particularly important subclass of metafinite structures are the so-called \emph{$\bbR$-structures}, which extend finite structures with the real arithmetic on the second sort. 
\emph{Blum-Shub-Smale machines} (BSS machines for short) \cite{blum1989} are essentially register machines with registers that can store arbitrary real numbers and compute rational functions over reals in a single time step. 
Interestingly, Boolean languages which are decidable by a non-deterministic polynomial-time BSS machine coincide with those languages which are PTIME-reducible to the true existential sentences of real arithmetic (i.e., the complexity class $\exists\bbR$) \cite{BurgisserC06,SchaeferS17}.

Recent works have established fascinating connections between second-order logics over $\bbR$-structures, complexity classes using the BSS-model of computation, and logics using probabilistic team semantics. 
In \cite{HannulaKBV20}, Hannula et al.\ establish that the expressivity and complexity of probabilistic independence logic coincide with a particular fragment of existential second-order logic over $\bbR$-structures and NP on BSS-machines. 
In \cite{HannulaV22}, Hannula and Virtema focus on probabilistic inclusion logic, which is shown to be tractable (when restricted to Boolean inputs), and relate it to linear programming.

In this paper, we focus on the expressivity and model checking complexity of probabilistic team-based logics that have access to Boolean negation. 
We also study the connections between probabilistic independence logic and a logic called $\FOPT$, which is defined via a computationally simpler probabilistic semantics~\cite{HHK22}. 
The logic $\FOPT$ is the probabilistic variant of a certain team-based logic that can define exactly those dependencies that are first-order definable \cite{kontinen_yang_2022}. 
We also introduce novel entropy atoms and relate the extension of first-order logic with these atoms to probabilistic independence logic.
This version of the paper includes the proofs omitted from the conference version \cite{jelia23}. 

See Figure \ref{fig:landscape} for our expressivity results and Table \ref{tab:results} for our complexity results.

\begin{figure}[t]
	\centering
\resizebox{\linewidth}{!}{\begin{tikzpicture}
		\node at (-1,4) {\underline{formulas}:};

		\node (SOR+xlog) at (0,3.5) {$\SO_{\mathbb R}(+,\times,\log)$};
		\node (FOindneg) at (1.7,2.5) {$\FO(\cpind,\sim)=\SO_{\mathbb R}(+,\times)$};
		\node (FOH) at (-1.7,2.5) {$\FO(\eH)$};
		\node (FOind) at (0,1) {$\FO(\cpind)$};
		
		\node (FOapprox) at (0,0) {$\FO(\approx)$};
		\node (FOPT) at (1.7,0) {$\FOPT$};		

		\node (FO) at (1,-1) {$\FO$};

		\node at (5,4) {\underline{sentences}:};

		\node (sSOR+xlog) at (6,3.5) {$\SO_{\mathbb R}(+,\times,\log)$};
		\node[label={0:{\sffamily\footnotesize [Thm.~\ref{thm_indSO}]}}] (sFOindneg) at (7.5,2.5) {$\FO(\cpind,\sim)=\SO_{\mathbb R}(+,\times)$};
		\node (sFOH) at (4.5,2.5) {$\FO(\eH)$};
		\node (sFOind) at (6,1) {$\FO(\cpind)$};		
		\node (sFOapprox) at (6,0) {$\FO(\approx)$};
		\node[label={0:{\sffamily\footnotesize [Thm.~\ref{thm:FOPT}]}}] (sFO) at (6,-1) {$\FO=\FOPT$};

		\foreach \f/\t in {FOind/FOH,FOH/SOR+xlog,FOindneg/SOR+xlog,sFOapprox/sFOind,sFOind/sFOH,sFOind/sFOindneg,sFOH/sSOR+xlog,sFOindneg/sSOR+xlog}{
			\path[->] (\f) edge (\t);
		}
		
		\foreach \f/\t in {FO/FOapprox,FO/FOPT,FOapprox/FOind,FOind/FOindneg,FOPT/FOindneg,sFO/sFOapprox}{
			\path[->>] (\f) edge (\t);
		}
		
		\path (FOPT) edge node[midway, above, sloped, yshift=-.5mm] {\sffamily\footnotesize Thm.~\ref{thm:openforms-fopt-fpind-neg}} (FOindneg);
		\path (FOapprox) edge node[xshift=-1.5mm,anchor=east] {{\sffamily\footnotesize \cite[Thm.~10]{HannulaHKKV19}}} (FOind);
		\path (FOind) edge node[midway, above, sloped, yshift=-.5mm] {{\sffamily\footnotesize Cor.~\ref{cor:FO-ind-eh}}} (FOH);
		\path (FOind) edge node[midway, above, sloped, yshift=-.5mm] {{\sffamily\footnotesize Prop.~\ref{prop:FO-ind-strict-FO-ind-sim}}} (FOindneg);
		\path (FOH) edge node[midway, below, sloped, xshift=.5mm] {{\sffamily\footnotesize Thm.~\ref{thm:H-ESO-H-neg-SO}}} (SOR+xlog);		
		
		\foreach \l/\x/\y in {{AC^0}/11/-1,P/11/0,{\exists\mathbb R}/11/1}{
			\node at (\x,\y) {$\mathsf\l$};
		}
		
	\end{tikzpicture}}
	\caption{Landscape of relevant logics as well as relation to some complexity classes. Note that for the complexity classes, finite ordered structures are required. Single arrows indicate inclusions and double arrows indicate strict inclusions.}\label{fig:landscape}
\end{figure}
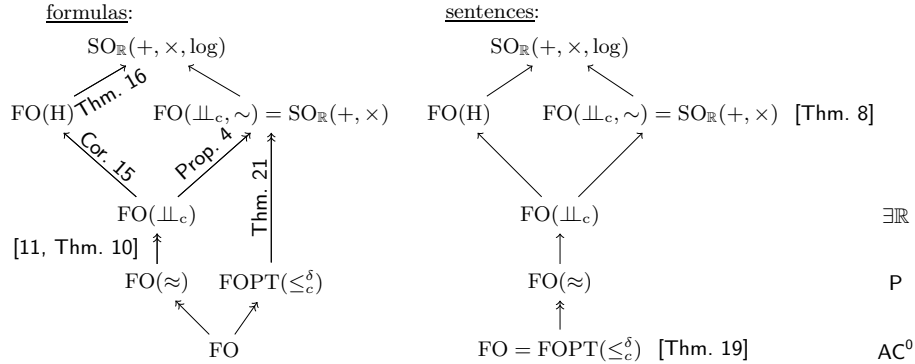

\begin{table}
	\centering
	\resizebox{\linewidth}{!}{%
	\begin{tabular}{cccc}\toprule
	Logic & $\MC$ for sentences& $\SAT$ & $\VAL$\\\midrule
	$\FOPT$ & 
	$\PSPACE$ {\small(Cor.~\ref{cor:sent-mcfopt-pspace})}&
	$\RE$ \cite[Thm.~5.2]{HHK22}&$\coRE$ \cite[Thm.~5.2]{HHK22}
	\\
	$\FO(\cpind)$ & 
	$\in\EXPSPACE$ and $\NEXPTIME$-hard {\small(Thm.~\ref{thm:MCFO-ind})}& 
	$\RE$ \small{(Thm.~\ref{thm:ind-sim-val-sat})} & 
	$\coRE$ {\small(Thm.~\ref{thm:ind-sim-val-sat})}
	\\
	$\FO(\sim)$ &
	$\AEXPTIME[\poly]$ \cite[Prop.~5.16, Lem.~5.21]{DBLP:phd/dnb/Luck20}& 
	$\RE$~\cite[Thm.~5.6]{DBLP:phd/dnb/Luck20} & 
	$\coRE$~\cite[Thm.~5.6]{DBLP:phd/dnb/Luck20}
	\\
	$\probinclogic$ & 
	$\in\EXPTIME$, $\PSPACE$-hard {\small(Thm.~\ref{thm:MC-inc})}&
	$\RE$ \small{(Thm.~\ref{thm:ind-sim-val-sat})} & 
	$\coRE$ {\small(Thm.~\ref{thm:ind-sim-val-sat})}
	\\
	$\FO(\sim,\cpind)$ & 
	$\in\threeEXPSPACE$, $\AEXPTIME[\poly]$-hard {\small(Thm.~\ref{thm:sent-mc-sim-ind-expspace-aexp-poly})} & 
	$\RE$ \small{(Thm.~\ref{thm:ind-sim-val-sat})} & 
	$\coRE$ {\small(Thm.~\ref{thm:ind-sim-val-sat})}
	\\\bottomrule		
	\end{tabular}}
	
	\caption{Overview of our results. Unless otherwise noted, the results are completeness results. Satisfiability and Validity are considered for finite structures.}\label{tab:results}
\end{table}

\section{Preliminaries}

We assume the reader is familiar with the basics in complexity theory~\cite{0072413}.
In this work, we will encounter complexity classes $\PSPACE$, $\EXPTIME$, $\NEXPTIME$, $\EXPSPACE$ and the class $\AEXPTIME[\poly]$ together with the notion of completeness under the usual polynomial time many to one reductions. 
A bit more formally for the latter complexity class which is more uncommon than the others, $\AEXPTIME[\poly]$ consists of all languages that can be decided by alternating Turing machines within an exponential runtime of $O(2^{n^{O(1)}})$ and polynomially many alternations between universal and existential states. 
There exist problems in propositional team logic with generalized dependence atoms that are complete for this class~\cite{DBLP:journals/corr/HannulaKLV16}. 
It is also known that truth evaluation of alternating dependency quantified boolean formulae (ADQBF) is complete for this class~\cite{DBLP:journals/corr/HannulaKLV16}. 



\subsection{Probabilistic team semantics}

We denote first-order variables by $x,y,z$ and tuples of first-order variables by $\tuple x,\tuple y,\tuple z$. For the length of the tuple $\tuple x$, we write $|\tuple x|$. The set of variables that appear in the tuple $\tuple x$ is denoted by $\var(\tuple x)$. A vocabulary $\tau$ is a finite set of relation, function, and constant symbols, denoted by $R$, $f$, and $c$, respectively. Each relation symbol $R$ and function symbol $f$ has a prescribed arity, denoted by $\ar(R)$ and $\ar(f)$.

Let $\tau$ be a finite relational vocabulary such that $\{=\}\subseteq\tau$. For a finite $\tau$-structure $\A$ and a finite set of variables  $D$, an \textit{assignment} of $\A$ for $D$ is a function $s\colon D\to A$. A \textit{team} $X$ of $\A$ over $D$ is a finite set of assignments $s\colon D\to A$.

A \textit{probabilistic team} $\X$ is a function $\X\colon X\to\mathbb{R}_{\geq 0}$, where $\mathbb{R}_{\geq 0}$ is the set of non-negative real numbers. The value $\X(s)$ is called the \textit{weight} of assignment $s$. Since zero-weights are allowed, we may, when useful, assume that $X$ is maximal, i.e., it contains all assignments $s\colon D\to A$. The \textit{support} of $\X$ is defined as $\supp(\X)\coloneqq\{s\in X\mid\X(s)\neq 0\}$. A team $\X$ is \textit{nonempty} if $\supp(\X)\neq\varnothing$.

These teams are called probabilistic because we usually consider teams that are probability distributions, i.e., functions $\X\colon X\to\mathbb{R}_{\geq 0}$ for which $\sum_{s\in X}\X(s)=1$.\footnote{In some sources, the term probabilistic team only refers to teams that are distributions, and the functions  $\X\colon X\to\mathbb{R}_{\geq 0}$ that are not distributions are called \textit{weighted teams}.}  
In this setting, the weight of an assignment can be thought of as the probability that the values of the variables are as in the assignment.
If $\X$ is a probability distribution, we also write $\X\colon X\to [0,1]$. 

For a set of variables  $V$, the restriction of the assignment $s$ to $V$ is denoted by $s\restriction{V}$. 
The \textit{restriction of a team} $X$ to $V$ is $X\restriction{V}=\{s\restriction{V}\mid s\in X\}$, and the \textit{restriction of a probabilistic team} $\X$ to $V$ is $\X\restriction{V}\colon X\restriction{V}\to \mathbb{R}_{\geq 0}$ where
\[
(\X\restriction{V})(s)=\sum_{\substack{s'\restriction{V}=s,\\ s'\in X}}\X(s').
\]

If $\phi$ is a first-order formula, then $\X_\phi$ is the restriction of the team $\X$ to those assignments in $X$ that satisfy the formula $\phi$. The weight $|\X_\phi|$ is defined analogously as the sum of the weights of the assignments in $X$ that satisfy $\phi$, e.g.,
\[
|\X_{\tuple x=\tuple a}|=\sum_{\substack{s\in X,\\ s(\tuple x)=\tuple a}}\X(s).
\]

For a variable $x$ and $a\in A$, we denote by $s(a/x)$, the modified assignment $s(a/x)\colon D\cup\{x\}\to A$ such that $s(a/x)(y)=a$ if $y=x$, and $s(a/x)(y)=s(y)$ otherwise. For a set $B\subseteq A$, the modified team $X(B/x)$ is defined as the set $X(B/x)\coloneqq\{s(a/x)\mid a\in B, s\in X\}$. 

Let $\X\colon X\to\mathbb{R}_{\geq 0}$ be any probabilistic team. Then the probabilistic team $\X(B/x)$ is a function $\X(B/x)\colon X(B/x)\to\mathbb{R}_{\geq 0}$ defined as
\[
\X(B/x)(s(a/x))=\sum_{\substack{t\in X,\\ t(a/x)=s(a/x)}}\X(t)\cdot\frac{1}{|B|}.
\]
If $x$ is a fresh variable, the summation can be dropped and the right-hand side of the equation becomes  $\X(s)\cdot\frac{1}{|B|}$. For singletons $B=\{a\}$, we write $X(a/x)$ and $\X(a/x)$ instead of $X(\{a\}/x)$ and $\X(\{a\}/x)$. 

Let then $\X\colon X\to[0,1]$ be a distribution. Denote by $p_B$ the set of all probability distributions $d\colon B\to[0,1]$, and let $F$ be a function $F\colon X\to p_B$. 
Then the probabilistic team $\X(F/x)$ is a function $\X(F/x)\colon X(B/x)\to[0,1]$ defined as
\[
\X(F/x)(s(a/x))=\sum_{\substack{t\in X,\\ t(a/x)=s(a/x)}}\X(t)\cdot F(t)(a)
\]
for all $a\in B$ and $s\in X$. If $x$ is a fresh variable, the summation can again be dropped and the right-hand side of the equation becomes $\X(s)\cdot F(s)(a)$.

Let $\X\colon X\to[0,1]$ and $\mathbb{Y}\colon Y\to[0,1]$ be probabilistic teams with common variable and value domains, and let $k\in[0,1]$. The $k$-scaled union of $\X$ and $\mathbb{Y}$, denoted by $\mathbb{X}\sqcup_k\mathbb{Y}$, is the probabilistic team $\mathbb{X}\sqcup_k\mathbb{Y}\colon Y\to[0,1]$ defined as
\[
\mathbb{X}\sqcup_k\mathbb{Y}(s)\coloneqq
\begin{cases}
k\cdot\X(s)+(1-k)\cdot\mathbb{Y}(s)\qquad &\text{if } s\in X\cap Y,\\
k\cdot\X(s) &\text{if }s\in X\setminus Y,\\
(1-k)\cdot\mathbb{Y}(s) &\text{if }s\in Y\setminus X.
\end{cases}
\]

\section{Probabilistic independence logic with Boolean negation}

In this section, we define probabilistic independence logic with Boolean negation, denoted by $\FO(\cpind,\sim)$. The logic extends first-order logic with \textit{probabilistic independence atom} $\pci{\tuple x}{\tuple y}{\tuple z}$ which states that the tuples $\tuple y$ and $\tuple z$ are independent given the tuple $\tuple x$.
The syntax for the logic $\FO(\cpind,\sim)$ over a vocabulary $\tau$ is as follows:
\[
\phi\Coloneqq  R(\tuple x)\mid \neg R(\tuple x) \mid \pci{\tuple x}{\tuple y}{\tuple z} \mid \bn\phi \mid (\phi\wedge\phi) \mid (\phi\lor\phi) \mid \exists x\phi \mid \forall x\phi,
\]
where $x$ is a first-order variable, $\tuple x$, $\tuple y$, and $\tuple z$ are tuples of first-order variables, and $R\in\tau$.

Let $\psi$ be a first-order formula. We denote by $\psi^\neg$ the formula which is obtained from $\neg\psi$ by pushing the negation in front of atomic formulas. We also use the shorthand notations $\psi \to \phi \dfn (\psi^\neg \lor (\psi \land \phi))$ and $\psi \leftrightarrow \phi\coloneqq\psi \to \phi\wedge\phi \to \psi$.

Let $\X\colon X\to[0,1]$ be a probability distribution. The semantics for the logic is defined as follows:
\begin{itemize}
\item[] $\A\models_{\X}R(\tuple x)$ iff $\A\models_{s}R(\tuple x)$ for all $s\in \supp(\X)$.
\item[] $\A\models_{\X}\neg R(\tuple x)$ iff $\A\models_{s}\neg R(\tuple x)$ for all $s\in \supp(\X)$.
\item[] $\A\models_{\X}\pci{\tuple x}{\tuple y}{\tuple z}$ iff $|\X_{\tuple x\tuple y=s(\tuple x\tuple y)}|\cdot|\X_{\tuple x\tuple z=s(\tuple x\tuple z)}|=|\X_{\tuple x\tuple y\tuple z=s(\tuple x\tuple y\tuple z)}|\cdot|\X_{\tuple x=s(\tuple x)}|$ for all $s\colon\var(\tuple x\tuple y\tuple z)\to A$.
\item[] $\A\models_{\X}\bn\phi$ iff $\A\not\models_{\X}\phi$.
\item[] $\A\models_{\X}\phi\wedge\psi$ iff $\A\models_{\X}\phi$ and $\A\models_{\X}\psi$.
\item[] $\A\models_{\X}\phi\lor\psi$ iff $\A\models_{\mathbb{Y}}\phi$ and $\A\models_{\mathbb{Z}}\psi$ for some $\mathbb{Y},\mathbb{Z},k$ such that $\mathbb{Y}\sqcup_k\mathbb{Z}=\X$.
\item[] $\A\models_{\X}\exists x\phi$ iff $\A\models_{\X(F/x)}\phi$ for some $F\colon X\to p_A$.
\item[] $\A\models_{\X}\forall x\phi$ iff $\A\models_{\X(A/x)}\phi$. 
\end{itemize}

The satisfaction relation $\models_s$ above refers to the Tarski semantics of first-order logic. For a sentence $\phi$, we write $\A\models\phi$ if $\A\models_{\X_\emptyset}\phi$, where $\X_\emptyset$ is the distribution that maps the empty assignment to 1.

The logic also has the following useful property called \textit{locality}. Denote by $\Fr(\phi)$ the set of the free variables of a formula $\phi$.
\begin{proposition}[\textbf{Locality}, {\cite[Prop.~12]{DHKMV18}}]\label{locality2}
Let $\phi$ be any $\FO(\cpind,\sim)[\tau]$-formula. Then for any set of variables $V$, any $\tau$-structure $\A$, and any probabilistic team $\X\colon X\to[0,1]$ such that $\Fr(\phi)\subseteq V\subseteq D$, 
\[
\A\models_{\X}\phi \iff \A\models_{\X\restriction{V}}\phi.
\]
\end{proposition}

In addition to probabilistic conditional independence atoms, we may also consider other atoms. If $\tuple x$ and $ \tuple y$ are tuples of variables, then $\dep(\tuple x,\tuple y)$ is a \textit{dependence atom}. If $\tuple x$ and $ \tuple y$ are also of the same length, $\tuple x\approx\tuple y$ is a \textit{marginal identity atom}. The semantics for these atoms are defined as follows:
\begin{itemize}
\item[] $\A\models_{\X}\dep(\tuple x,\tuple y)$ iff for all $s,s'\in\supp(\X)$, $s(\tuple x)=s'(\tuple x)$ implies $s(\tuple y)=s'(\tuple y)$,
\item[] $\A\models_{\X}\tuple x\approx\tuple y$ iff $|\X_{\tuple x=\tuple a}|=|\X_{\tuple y=\tuple a}|$ for all $\tuple a\in A^{|\tuple x|}$.
\end{itemize}
We write $\FO(\dep(\cdot))$ and $\FO(\approx)$ for first-order logic with dependence atoms or marginal identity atoms, respectively. Analogously, for $C\subseteq\{\dep(\cdot),\approx,\cpind,\sim\}$, we write $\FO(C)$ for the logic with access to the atoms (or the Boolean negation) from $C$.

For two logics $L$ and $L'$ over probabilistic team semantics, we write $L\leq L'$ if for any formula $\phi\in L$, there is a formula $\psi\in L'$ such that $\A\models_\X\phi \iff \A\models_\X\psi$ for all $\A$ and $\X$. The equality $\equiv$ and strict inequality $<$ are defined from the above relation in the usual way. The next two propositions follow from the fact that dependence atoms and marginal identity atoms can be expressed with probabilistic independence atoms.
\begin{proposition}[{\cite[Prop.~24]{DurandHKMV18}}]\label{prop_depcpind}
$\FO(\dep(\cdot))\leq \FO(\cpind)$.
\end{proposition}
\begin{proposition}[{\cite[Thm.~10]{HannulaHKKV19}}]\label{prop_approxcpind}
$\FO(\approx)\leq \FO(\cpind)$.
\end{proposition}

On the other hand, omitting the Boolean negation strictly decreases the expressivity:
\begin{proposition}\label{prop:FO-ind-strict-FO-ind-sim}
$\FO(\cpind) < \FO(\cpind,\sim)$.
\end{proposition}

\begin{proof}
By Theorems 4.1 and 6.5 of \cite{HannulaKBV20}, over a fixed universe size, any open formula of $\FO(\cpind)$ defines a closed subset of $\mathbb{R}^n$ for a suitable $n$ depending on the size of the universe and the number of free variables. 
Now, clearly, this cannot be true for all of the formulas of $\FO(\cpind,\sim)$ as it contains the Boolean negation, e.g., the formula $\sim x \pind_y z$.
\end{proof}

\section{Metafinite logics}

In this section, we consider logics over $\mathbb{R}$-structures. These structures extend finite relational structures with real numbers $\mathbb{R}$ as a second domain and add functions that map tuples from the finite domain to $\mathbb{R}$.

\begin{definition}[$\mathbb{R}$-structures]
Let $\tau$ and $\sigma$ be finite vocabularies such that $\tau$ is relational and $\sigma$ is functional. 
An \emph{$\mathbb{R}$-structure of vocabulary $\tau\cup\sigma$} is a tuple $\mathcal{A}=(A,\mathbb{R},F)$ where the reduct of $\mathcal{A}$ to $\tau$ is a finite relational structure, and $F$ is a set that contains functions $f^{\mathcal{A}}\colon A^{\ar(f)}\to \mathbb{R}$ for each function symbol $f\in\sigma$. Additionally,
\begin{enumerate}[(i)]
\item for any $S\subseteq\mathbb{R}$, if each $f^{\mathcal{A}}$ is a function from $A^{\ar(f)}$ to $S$, $\mathcal{A}$ is called an \emph{$S$-structure},
\item if each $f^{\mathcal{A}}$ is a distribution, $\mathcal{A}$ is called a \emph{$d[0,1]$-structure}.
\end{enumerate}  
\end{definition}

Next, we will define certain metafinite logics which are variants of functional second-order logic with numerical terms. The numerical $\sigma$-terms $i$ are defined as follows:
\[
i\ddfn f(\tuple x) \mid i\times i \mid i+i \mid \SUM_{\tuple y}i \mid \log i,
\]
where $f\in\sigma$ and $\tuple x$ and $\tuple y$ are first-order variables such that $|\tuple x|=\ar(f)$. The interpretation of a numerical term $i$ in the structure $\A$ under an assignment $s$ is denoted by $[i]_s^\A$. We define
\[
[\SUM_{\tuple y}i]_s^\A:=\sum_{\tuple a\in A^{|\tuple y|}}[i]_{s(\tuple a/\tuple y)}^{\A}.
\]
The interpretations of the rest of the numerical terms are defined in the obvious way.

Suppose that $\{=\}\subseteq\tau$, and let $O\subseteq\{+,\times,\SUM,\log\}$. The syntax for the logic $\SO_{\mathbb{R}}(O)$ is defined as follows:
\[
\phi\Coloneqq i = j \mid \neg i = j \mid R(\tuple x)\mid \neg R(\tuple x) \mid (\phi\wedge\phi) \mid (\phi\lor\phi) \mid \exists x\phi \mid \forall x\phi \mid \exists f\psi \mid \forall f\psi,
\]
where $i$ and $j$ are numerical $\sigma$-terms constructed using operations from $O$, $R\in\tau$, $x$, $y$, and $\tuple x$ are first-order variables, $f$ is a function variable, and $\psi$ is a $\tau\cup\sigma\cup\{f\}$-formula of $\SO_{\mathbb{R}}(O)$. 

The semantics of $\SO_{\mathbb{R}}(O)$ is defined via $\mathbb{R}$-structures and assignments analogous to first-order logic, except for the interpretations of function variables $f$, which range over functions $A^{\ar(f)}\to \mathbb{R}$. For any $S\subseteq\mathbb{R}$, we define $\SO_{S}(O)$ as the variant of $\SO_{\mathbb{R}}(O)$, where the quantification of function variables ranges over $A^{\ar(f)}\to S$. If the quantification of function variables is restricted to distributions, the resulting logic is denoted by $\SO_{d[0,1]}(O)$. The existential fragment, in which universal quantification over function variables is not allowed, is denoted by $\ESO_{\mathbb{R}}(O)$.

For metafinite logics $L$ and $L'$, we define expressivity comparison relations $L\leq L'$, $L\equiv L'$, and $L< L'$ in the usual way, see e.g. \cite{HannulaKBV20}.

\begin{proposition}\label{SUM+_prop}
$\SO_{\mathbb{R}}(\SUM, \times)\equiv\SO_{\mathbb{R}}(+,\times)$.
\end{proposition}
\begin{proof}
First, note that since the constants 0 and 1 are definable in both logics, we may use them when needed. To show that $\SO_{\mathbb{R}}(\SUM,\times)\leq\SO_{\mathbb{R}}(+,\times)$, it suffices to show that any numerical identity $f(\tuple x)=\SUM_{\tuple y}g(\tuple x,\tuple y)$ can also be expressed in $\SO_{\mathbb{R}}(+,\times)$. Suppose that $|\tuple y|=n$. Since the domain of $\A$ is finite, we may assume that it is linearly ordered: a linear order $\leq_{\text{fin}}$ can be defined with an existentially quantified binary function variable $f$ such that the formulas $f(x,y)=1$ and $f(x,y)=0$ correspond to $x\leq_{\text{fin}}y$ and $x\not\leq_{\text{fin}}y$, respectively.
Then, without loss of generality, we may assume that we have an $n$-ary successor function $S$ defined by the lexicographic order induced by the linear order. Thus, we can existentially quantify a function variable $h$ such that 
\[
\forall\tuple x\tuple z(h(\tuple x,\bold{min})=g(\tuple x,\bold{min})\wedge h(\tuple x,S(\tuple z))=h(\tuple x,\tuple z)+g(\tuple x,S(\tuple z)).
\]
Then $f(\tuple x)=h(\tuple x,\bold{max})$ is as wanted.

To show that $\SO_{\mathbb{R}}(+,\times)\leq\SO_{\mathbb{R}}(\SUM,\times)$, we show that any numerical identity $f(\tuple x\tuple y)=i(\tuple x)+j(\tuple y)$ can be expressed in $\SO_{\mathbb{R}}(\SUM,\times)$. We can existentially quantify a function variable $g$ such that
\begin{align*}
g(\tuple x\tuple y,min)=&i(\tuple x)\wedge g(\tuple x\tuple y,max)=j(\tuple y)\\
&\wedge\forall z( (\neg z=min\wedge \neg z=max)\rightarrow g(\tuple u\tuple v,z)=0).
\end{align*}
Then $f(\tuple x\tuple y)=\SUM_zg(\tuple x\tuple y,z)$ is as wanted. Note that since no universal quantification over function variables was used, the proposition also holds for existential fragments, i.e., $\ESO_{\mathbb{R}}(\SUM,\times)\equiv\ESO_{\mathbb{R}}(+,\times)$.
\end{proof}


\begin{proposition}\label{SOprop}
$\SO_{d[0,1]}(\SUM,\times)\equiv\SO_{\mathbb{R}}(+,\times)$.
\end{proposition}
\begin{proof}
Since $1$ is definable in $\SO_{\mathbb{R}}(\SUM,\times)$ and the formula $\SUM_{\tuple x}f(\tuple x)=1$ states that $f$ is a probability distribution, we have that $\SO_{d[0,1]}(\SUM,\times)\leq \SO_{\mathbb{R}}(\SUM,\times)\equiv\SO_{\mathbb{R}}(+,\times)$.

Next, we show that \[\SO_{\mathbb{R}}(+,\times)\leq\SO_{\mathbb{R}_{\geq 0}}(+,\times)\leq\SO_{[0,1]}(+,\times)\leq\SO_{d[0,1]}(\SUM,\times).\] 

To show that $\SO_{\mathbb{R}}(+,\times)\leq\SO_{\mathbb{R}_{\geq 0}}(+,\times)$, let $\phi\in\SO_{\mathbb{R}}(+,\times)$. Note that any function $f\colon A^{\ar(f)}\to\mathbb{R}$ can be expressed as $f_+-f_-$, where $f_+$ and $f_-$ are functions  $A^{\ar(f)}\to\mathbb{R}_{\geq 0}$ such that $f_+(\tuple x)=f(\tuple x)\cdot\chi_{\mathbb{R}_{\geq 0}}(f(\tuple x))$ and $f_-(\tuple x)=f(\tuple x)\cdot\chi_{\mathbb{R}\setminus\mathbb{R}_{\geq 0}}(f(\tuple x))$, where $\chi_S\colon\mathbb{R}\to\{0,1\}$ is the characteristic function of $S\subseteq\mathbb{R}$. Since numerical terms $i(\tuple x)-j(\tuple x)$ can clearly be expressed in $\SO_{\mathbb{R}}(+,\times)$, it suffices to modify $\phi$ as follows: for all quantified function variables $f$, replace each appearance of term $f(\tuple x)$ with $f_+(\tuple x)-f_-(\tuple x)$ and instead of $f$, quantify two function variables $f_+$ and $f_-$.

To show that $\SO_{\mathbb{R}_{\geq 0}}(+,\times)\leq\SO_{[0,1]}(+,\times)$, let $\phi\in\SO_{\mathbb{R}_{\geq 0}}(+,\times)$. Note that any positive real number can be written as a ratio $x/(1-x)$, where $x\in[0,1)$. Since numerical terms of the form $i(\tuple x)/(1-i(\tuple x))$ can clearly be expressed in $\SO_{d[0,1]}(+,\times)$, it suffices to modify $\phi$ as follows: for all quantified function variables $f$, replace each appearance of term $f(\tuple x)$ with $f^*(\tuple x)/(1-f^*(\tuple x))$ and instead of $f$, quantify a function variable $f^*$ such that $f^*(\tuple x)\neq 1$ for all $\tuple x$.

Lastly, to show that $\SO_{[0,1]}(+,\times)\leq\SO_{d[0,1]}(\SUM,\times)$, it suffices to see that for any $\phi\in\SO_{[0,1]}(+,\times)$, we can compress each function term into a fraction of size $1/n^k$, where $n$ is the size of the finite domain and $k$ the maximal arity of any function variable appearing in $\phi$. 
We omit the proof, since it is essentially the same as the one for Lemma 6.4 in \cite{HannulaKBV20}.
\end{proof}

\section{Equi-expressivity of $\FO(\cpind,\sim)$ and $\SO_{\mathbb{R}}(+,\times)$}

In this section, we show that the expressivity of probabilistic independence logic with the Boolean negation coincides with full second-order logic over $\mathbb{R}$-structures.

\begin{theorem}\label{thm_indSO}
$\FO(\cpind,\sim)\equiv \SO_{\mathbb{R}}(+,\times)$.
\end{theorem}




We first show that $\FO(\cpind,\sim)\leq\SO_{\mathbb{R}}(+,\times)$. Note that by Proposition \ref{SOprop}, we have $\SO_{d[0,1]}(\SUM,\times)\equiv\SO_{\mathbb{R}}(+,\times)$, so it suffices to show that $\FO(\cpind,\sim)\leq\SO_{d[0,1]}(\SUM,\times)$. We may assume that every independence atom is in the form $\pci{\tuple x}{\tuple y}{\tuple z}$ or $\pci{\tuple x}{\tuple y}{\tuple y}$ where $\tuple x,\tuple y,$ and $\tuple z$ are pairwise disjoint tuples. \cite[Lemma 25]{DHKMV18}

\begin{theorem}\label{thm_ind<SO}
Let formula $\phi(\tuple v)\in\FO(\cpind,\sim)$ be such that its free-variables are from $\tuple v=(v_1,\dots,v_k)$. Then there is a formula $\psi_{\phi}(f)\in\SO_{d[0,1]}(\SUM,\times)$ with exactly one free function variable such that for all structures $\A$ and all probabilistic teams $\X\colon X\to[0,1]$, $\A\models_\X\phi(\tuple v)$ if and only if $(\A,f_\X)\models\psi_\phi(f)$, where $f_\X\colon A^k\to[0,1]$ is a probability distribution such that $f_\X(s(\tuple v))=\X(s)$ for all $s\in X$.
\end{theorem}

\begin{proof}
Define the formula $\psi_\phi(f)$ as follows:
\begin{enumerate}[1.]
\item If $\phi(\tuple v)=R(v_{i_1},\dots,v_{i_l})$, where $1\leq i_1,\dots,i_l\leq k$, then $\psi_\phi(f)\coloneqq\forall\tuple v(f(\tuple v)=0\lor R(v_{i_1},\dots,v_{i_l}))$.
\item If $\phi(\tuple v)=\neg R(v_{i_1},\dots,v_{i_l})$, where $1\leq i_1,\dots,i_l\leq k$, then $\psi_\phi(f)\coloneqq\forall\tuple v(f(\tuple v)=0\lor \neg R(v_{i_1},\dots,v_{i_l}))$.
\item If $\phi(\tuple v)=\pci{\tuple v_0}{\tuple v_1}{\tuple v_2}$, where $\tuple v_0,\tuple v_1,\tuple v_2$ are disjoint, then 
\begin{align*}
\psi_\phi(f)\coloneqq\forall\tuple v_0\tuple v_1\tuple v_2(&\SUM_{\tuple v\backslash(\tuple v_0\tuple v_1)} f(\tuple v)\times\SUM_{\tuple v\backslash(\tuple v_0\tuple v_2)} f(\tuple v)=\\&\SUM_{\tuple v\backslash(\tuple v_0\tuple v_1)} f(\tuple v)\times\SUM_{\tuple v\backslash\tuple v_0} f(\tuple v)).
\end{align*}
\item If $\phi(\tuple v)=\pci{\tuple v_0}{\tuple v_1}{\tuple v_1}$, where $\tuple v_0,\tuple v_1$ are disjoint, then 
\[
\psi_\phi(f)\coloneqq\forall\tuple v_0\tuple v_1(\SUM_{\tuple v\backslash(\tuple v_0\tuple v_1)} f(\tuple v)=0\lor\SUM_{\tuple v\backslash(\tuple v_0\tuple v_1)} f(\tuple v)=\SUM_{\tuple v\backslash\tuple v_0} f(\tuple v)).
\]
\item If $\phi(\tuple v)=\bn\phi_0(\tuple v)$, then $\psi_\phi(f)\coloneqq\psi_{\phi_0}^\neg(f)$, where $\psi_{\phi_0}^\neg$ is obtained from $\neg\psi_{\phi_0}$ by pushing the negation in front of atomic formulas.
\item If $\phi(\tuple v)=\phi_0(\tuple v)\wedge\phi_1(\tuple v)$, then $\psi_\phi(f)\coloneqq\psi_{\phi_0}(f)\wedge\psi_{\phi_1}(f)$.
\item If $\phi(\tuple v)=\phi_0(\tuple v)\lor\phi_1(\tuple v)$, then 
\begin{align*}
\psi_{\phi}(f)&\coloneqq\psi_{\phi_0}(f)\lor\psi_{\phi_1}(f)\\
&\lor(\exists g_0 g_1 g_2 g_3 (\forall\tuple v\forall x(x=l\lor x=r\lor (g_0(x)=0\wedge g_3(\tuple v,x)=0))\\
&\wedge\forall\tuple v(g_3(\tuple v,l)=g_1(\tuple v)\times g_0(l)\wedge g_3(\tuple v,r)=g_2(\tuple v)\times g_0(r))\\
&\wedge\forall\tuple v(\SUM_x g_3(\tuple v,x)=f(\tuple v))\wedge\psi_{\phi_0}(g_1)\wedge\psi_{\phi_1}(g_2))).
\end{align*}
\item If $\phi(\tuple v)=\exists x\phi_0(\tuple v,x)$, then $\psi_{\phi}(f)\coloneqq\exists g(\forall\tuple v(\SUM_x g(\tuple v,x)=f(\tuple v))\wedge\psi_{\phi_0}(g))$.
\item If $\phi(\tuple v)=\exists x\phi_0(\tuple v,x)$, then 
\[
\psi_{\phi}(f)\coloneqq\exists g(\forall\tuple v(\forall x\forall y(g(\tuple v,x)=g(\tuple v,y))\wedge\SUM_x g(\tuple v,x)=f(\tuple v))\wedge\psi_{\phi_0}(g)).
\]
\end{enumerate}
Since the the above is essentially same as the translation in \cite[Theorem 14]{DHKMV18}, but extended with the Boolean negation (for which the claim follows directly from the semantical clauses), it is easy to show that $\psi_\phi(f)$ satisfies the claim. 
\end{proof}

We now show that $\SO_{\mathbb{R}}(+,\times)\leq\FO(\cpind,\sim,)$. By Propositions \ref{prop_approxcpind} and \ref{SOprop}, $\FO(\cpind,\sim,\approx)\equiv \FO(\cpind,\sim)$ and $\SO_{\mathbb{R}}(+,\times)\equiv\SO_{d[0,1]}(\SUM,\times)$, so it suffices to show that $\SO_{d[0,1]}(\SUM,\times)\leq\FO(\cpind,\sim,\approx)$.

Note that even though we consider $\SO_{d[0,1]}(\SUM,\times)$, where only distributions can be quantified, it may still happen that the interpretation of a numerical term does not belong to the unit interval. 
This may happen if we have a term of the form $\SUM_{\tuple x}i(\tuple y)$ where $\tuple x$ contains a variable that does not appear in $\tuple y$. 
Fortunately, for any formula containing such terms, there is an equivalent formula without them \cite[Lemma 19]{HannulaV22}. Thus, it suffices to consider formulas without such terms.

To prove that $\SO_{d[0,1]}(\SUM,\times)\leq\FO(\cpind,\sim,\approx)$, we construct a useful normal form for $\SO_{d[0,1]}(\SUM,\times)$-sentences. The following lemma is based on similar lemmas from \cite[Lemma, 16]{DHKMV18} and \cite[Lemma, 20]{HannulaV22}.
\begin{lemma}
Every formula $\phi\in\SO_{d[0,1]}(\SUM,\times)$ can be written in the form $\phi^*\coloneqq Q_1f_1\dots Q_nf_n\forall\tuple x\theta$, where $Q\in\{\exists, \forall\}$, $\theta$ is quantifier-free and such that all the numerical identity atoms are in the form $f_i(\tuple u\tuple v)=f_j(\tuple u)\times f_k(\tuple v)$ or $f_i(\tuple u)=\SUM_{\tuple v}f_j(\tuple u\tuple v)$ for distinct $f_i$,$f_j$,$f_k$ such that at most one of them is not quantified.
\end{lemma}
\begin{proof}
We begin by defining a formula $\theta_i$ for each numerical term $i(\tuple x)$ using fresh function symbols $f_i$.
\begin{enumerate}[1.]
\item If $i(\tuple u)=g(\tuple u)$ where $g$ is a function symbol, then $\theta_i$ is defined as $f_i(\tuple u)=\SUM_{\emptyset}g(\tuple u)$.
\item If $i(\tuple u\tuple v)=j(\tuple u)\times k(\tuple v)$, then $\theta_i$ is defined as $\theta_j\wedge\theta_k\wedge f_i(\tuple u\tuple v)=f_j(\tuple u)\times f_k(\tuple v)$.
\item If $i(\tuple u)=\SUM_{\tuple v}j(\tuple u\tuple v)$, then $\theta_i$ is defined as $\theta_j\wedge f_i(\tuple u)=\SUM_{\tuple v}f_j(\tuple u\tuple v)$.
\end{enumerate}
Then the formula $\phi^*$ is defined as follows:
\begin{enumerate}[1.]
\item If $\phi=i(\tuple u)=j(\tuple v)$, then $\phi^*\coloneqq\exists\tuple f(f_i(\tuple u)=f_j(\tuple v)\wedge\theta_i\wedge\theta_j)$ where $\tuple f$ consists of the function symbols $f_k$ for each subterm $k$ of $i$ or $j$. The negated case $\phi=\neg i(\tuple u)=j(\tuple v)$ is analogous; just add negation in front of $f_i(\tuple u)=f_j(\tuple v)$.
\item If $\phi$ is an atom or a negated atom (of the first sort), then  $\phi^*\coloneqq\phi$.
\item If $\phi=\psi_0\circ\psi_1$, where $\circ\in\{\wedge,\lor\}$ and $\psi^*_i=Q_1^i f_1^i\dots Q_{m_i}^if_{m_i}^i\forall\tuple x_i\theta_i$ for $i=0,1$, then $\phi^*\coloneqq Q_1^0 f_1^0\dots Q_{m_0}^0 f_{m_0}^0 Q_1^1 f_1^1\dots Q_{m_1}^1 f_{m_1}^1 \forall\tuple x_0\tuple x_1(\theta_0\circ\theta_1$).
\item If $\phi=\exists y\psi$, where $\psi^*=Q_1f_1\dots Q_mf_m\forall\tuple x\theta$, then 
\[
\phi^*\coloneqq\exists g Q_1f_1\dots Q_mf_m\forall\tuple x\forall\tuple y(g(y)=0\lor\theta).
\]
\item Let $\phi=\forall y\psi$, where $\psi^*=Q_1f_1\dots Q_mf_m\forall\tuple x\theta$. Let $f_{m+1},\dots,f_n$ list all of the free function variables in $\phi$. Then define
\begin{align*}
\phi^*\coloneqq &Q_1 f_1^*\dots Q_mf_m^*\exists f_{m+1}^*\dots\exists f_{n}^*\exists\tuple f_{id}\exists d\forall yy'\forall\tuple x(d(y)=d(y')\wedge\\
&\bigwedge_{i=m+1}^{n}(f_i^*(y,\tuple x)=f_i^*(y',\tuple x)\wedge\SUM_y f_i^*(y,\tuple x)=f_i(\tuple x))\wedge\\
&f_1^*(y,\tuple x)=f_1^*(y',\tuple x)\circ_1 (f_2^*(y,\tuple x)=f_2^*(y',\tuple x)\circ_2\dots\\
&\circ_{m-1}(f_m^*(y,\tuple x)=f_m^*(y',\tuple x)\circ_m\theta^*)\dots)),
\end{align*}
where each $f_i^*$, for $1\leq i\leq n$, is such that $\ar(f_i^*)=\ar(f_i)+1$, $\tuple f_{id}$ introduces a new function symbol for each multiplication in $\theta$,
\[
\circ_i\coloneqq
\begin{cases}
\wedge &\quad\text{if } Q_i=\exists,\\
\to &\quad\text{if } Q_i=\forall,
\end{cases}
\]
and the formula $\theta^*$ is obtained from $\theta$ by replacing all second sort identities $\alpha$ of the form $f_i(\tuple u\tuple v)=f_j(\tuple u)\times f_k(\tuple v)$ with
\[
f_\alpha(y,\tuple u\tuple v)=d(y)\times f_i^*(y,\tuple u\tuple v)\wedge f_\alpha(y,\tuple u\tuple v)=f_j^*(y,\tuple u)\times f_k^*(y,\tuple v)
\]
and $f_i(\tuple u)=\SUM_{\tuple v}f_j(\tuple u\tuple v)$ with $f_i^*(y,\tuple u)=\SUM_{\tuple v}f_j^*(y,\tuple u\tuple v)$.
\item If $\phi=Qf\psi$, where $Q\in\{\exists,\forall\}$ and $\psi^*=Q_1f_1\dots Q_mf_m\forall\tuple x\theta$, then $\phi^*\coloneqq Qf\psi^*$.
\end{enumerate}
It is straightforward to check that $\phi^*$ is as wanted. In (5), instead of quantifying for each $y$ a distribution $f_y$, we quantify a single distribution $f^*$ such that $f^*(y,\tuple x)=\frac{1}{|A|}\cdot f_y(\tuple x)$, where $A$ is the domain of our structure.
\end{proof}
\begin{lemma}\label{distr_lemma}
We use the abbreviations $\forall^*x\phi$ and $\phi\to^*\psi$ for the $\FO(\cpind,\sim,\approx)$-formulas $\bn\exists x\bn\phi$ and $\bn(\phi\wedge\bn\psi)$, respectively. Let $\phi_\exists\coloneqq\exists\tuple y(\pci{}{\tuple x}{\tuple y}\wedge\psi(\tuple x,\tuple y))$ and $\phi_\forall\coloneqq\forall^*\tuple y(\pci{}{\tuple x}{\tuple y}\to^*\psi(\tuple x,\tuple y))$ be $\FO(\cpind,\sim)$-formulas with free variables form $\tuple x=(x_1,\dots,x_n)$. Then for any structure $\mathcal{A}$ and probabilistic team $\X$ over $\{x_1,\dots,x_n\}$,
\begin{enumerate}[(i)]
\item $\mathcal{A}\models_{\X}\phi_\exists$ iff $\mathcal{A}\models_{\X(d/\tuple y)}\psi$ for some distribution $d\colon A^{|\tuple y|}\to[0,1]$,
\item $\mathcal{A}\models_{\X}\phi_\forall$ iff $\mathcal{A}\models_{\X(d/\tuple y)}\psi$ for all distributions $d\colon A^{|\tuple y|}\to[0,1]$.
\end{enumerate}
\end{lemma}
\begin{proof}
Let $\Y\coloneqq\X(\tuple F/\tuple y)$ for some sequence of functions $\tuple F=(F_1,\dots,F_{|\tuple y|})$ such that $F_i\colon X(A/y_1)\dots(A/y_i)\to p_A$. Now
\[
\mathcal{A}\models_{\Y}\pci{}{\tuple x}{\tuple y} \iff |\Y_{\tuple x\tuple y=s(\tuple x)\tuple a}|=|\Y_{\tuple x=s(\tuple x)}|\cdot|\Y_{\tuple y=\tuple a}| \text{ for all } s\in X,\tuple a\in A^{|\tuple y|}.
\]
Since the variables $\tuple y$ are fresh, the right-hand side becomes $\X(s)\cdot F_1(s)(a_1)\cdot\ldots\cdot F_{|\tuple y|}(s)(a_{|\tuple y|})=\X(s)\cdot|\Y_{\tuple y=\tuple a}|$ for all $s\in X,\tuple a\in A^{|\tuple y|}$, i.e., $\X(\tuple F/\tuple y)=\X(d/\tuple y)$ for some distribution $d\colon A^{|\tuple y|}\to[0,1]$. It is now straightforward to check that the two claims hold.
\end{proof}
\begin{theorem}
Let $\phi(p)\in\SO_{d[0,1]}(\SUM,\times)$ be a formula in the form $\phi^*\coloneqq Q_1f_1\dots Q_nf_n\forall\tuple x\theta$, where $Q\in\{\exists, \forall\}$, $\theta$ is quantifier-free and such that all the numerical identity atoms are in the form $f_i(\tuple u\tuple v)=f_j(\tuple u)\times f_k(\tuple v)$ or $f_i(\tuple u)=\SUM_{\tuple v}f_j(\tuple u\tuple v)$ for distinct $f_i$,$f_j$,$f_k$ from $\{f_1,\dots,f_n,p\}$. 
Then there is a formula $\Phi\in\FO(\cpind,\sim,\approx)$ such that for all structures $\mathcal{A}$ and probabilistic teams $\X\coloneqq p^\mathcal{A}$,
\[
\mathcal{A}\models_\X\Phi \text{ if and only if } (\mathcal{A},p)\models\phi.
\] 
\end{theorem}
\begin{proof}
Define 
\begin{align*}
\Phi\coloneqq\forall\tuple x&Q_1^*\tuple y_1(\pci{}{\tuple x}{\tuple y_1}\circ_1Q_2^*\tuple y_2(\pci{}{\tuple x\tuple y_1}{\tuple y_2}\circ_2Q_3^*\tuple y_3(\pci{}{\tuple x\tuple y_1\tuple y_2}{\tuple y_3}\circ_3\dots&\\
&Q_n^*\tuple y_n(\pci{}{\tuple x\tuple y_1\dots\tuple y_{n-1}}{\tuple y_n}\circ_n\Theta)\dots))),
\end{align*}
where 
$Q_i^*=\exists$ and $\circ_i=\wedge$, whenever $Q_i=\exists$ and $Q_i^*=\forall^*$ and $\circ_i=\to^*$, whenever $Q_i=\forall$.

By Lemma \ref{distr_lemma}, it suffices to show that for all distributions $f_1,\dots,f_n$, subsets $M\subseteq A^{|\tuple x|}$, and probabilistic teams $\Y\coloneqq\X(M/\tuple x)(f_1/\tuple y_1)\dots(f_n/\tuple y_n)$, we have
\[
\mathcal{A}\models_\Y\Theta \iff (\mathcal{A},p,f_1,\dots,f_n)\models\theta(\tuple a) \text{ for all } \tuple a\in M.
\]
The claim is shown by induction on the structure of the formula $\Theta$. 
\begin{enumerate}[1.]
\item If $\theta$ is an atom or a negated atom (of the first sort), then clearly we may let $\Theta\coloneqq\theta$.
\item Let $\theta=f_i(\tuple x_i)=f_j(\tuple x_j)\times f_k(\tuple x_k)$. Then define
\[
\Theta\coloneqq\exists\alpha\beta((\alpha=0\leftrightarrow\tuple x_i=\tuple y_i)\wedge(\beta=0\leftrightarrow\tuple x_j\tuple x_k=\tuple y_j\tuple y_k)\wedge\tuple x\alpha\approx\tuple x\beta).
\]
Assume first that $(\mathcal{A},p,f_1,\dots,f_n)\models\theta(\tuple a)$ for a given $\tuple a\in M$. Then $f_i(\tuple a_i)=f_j(\tuple a_j)\times f_k(\tuple a_k)$. Define functions $F_\alpha,F_\beta\colon Y\to\{0,1\}$ such that $F_\alpha(s)=0$ iff $s(\tuple x_i)=s(\tuple y_i)$, and $F_\beta(s)=0$ iff $s(\tuple x_j\tuple x_k)=s(\tuple y_j\tuple y_k)$. Let $\mathbb{Z}\coloneqq\Y(F_\alpha/\alpha)(F_\beta/\beta)$. It suffices to show that $\mathcal{A}\models_\mathbb{Z}\tuple x\alpha\approx\tuple x\beta$. Now, by the definition of $\mathbb{Z}$, we have $|\mathbb{Z}_{\tuple x\alpha=\tuple a0}|=|\mathbb{Z}_{\tuple x\tuple y_i=\tuple a\tuple a_i}|=|\mathbb{Y}_{\tuple x=\tuple a}|\cdot f_i(\tuple a_i)$ and $|\mathbb{Z}_{\tuple x\beta=\tuple a0}|=|\mathbb{Z}_{\tuple x\tuple y_j\tuple y_k=\tuple a\tuple a_j\tuple a_k}|=|\mathbb{Y}_{\tuple x=\tuple a}|\cdot f_j(\tuple a_j)\cdot f_k(\tuple a_k)$. Since $f_i(\tuple a_i)=f_j(\tuple a_j)\times f_k(\tuple a_k)$, we obtain $|\mathbb{Z}_{\tuple x\alpha=\tuple a0}|=|\mathbb{Z}_{\tuple x\beta=\tuple a0}|$ and $|\mathbb{Z}_{\tuple x\alpha=\tuple a1}|=|\mathbb{Y}_{\tuple x=\tuple a}|\cdot (1-f_i(\tuple a_i))=|\mathbb{Z}_{\tuple x\beta=\tuple a1}|$. Hence, $\mathcal{A}\models_\Y\Theta$.

Assume then that $\mathcal{A}\models_\Y\Theta$, and define $\mathbb{Z}$ as the extension of $\Y$ such that $\mathbb{Z}_{\alpha=0}=\mathbb{Z}_{\tuple x_i=\tuple y_i}$ and $\mathbb{Z}_{\beta=0}=\mathbb{Z}_{\tuple x_j\tuple x_k=\tuple y_j\tuple y_k}$. Then $|\mathbb{Y}_{\tuple x=\tuple a}|\cdot f_i(\tuple a_i)=|\mathbb{Z}_{\tuple x\tuple y_i=\tuple a\tuple a_i}|=|\mathbb{Z}_{\tuple x\tuple x_i=\tuple a\tuple y_i}|=|\mathbb{Z}_{\tuple x\alpha=\tuple a0}|=|\mathbb{Z}_{\tuple x\beta=\tuple a0}|=|\mathbb{Z}_{\tuple x\tuple x_j\tuple x_k=\tuple a\tuple y_j\tuple y_k}|=|\mathbb{Z}_{\tuple x\tuple y_j\tuple y_k=\tuple a\tuple a_j\tuple a_k}|=|\mathbb{Y}_{\tuple x=\tuple a}|\cdot f_j(\tuple a_j)\cdot f_k(\tuple a_k)$ for all $\tuple a\in M$. Hence, $(\mathcal{A},p,f_1,\dots,f_n)\models\theta(\tuple a)$ for all $\tuple a\in M$.

The negated case $\neg f_i(\tuple x_i)=f_j(\tuple x_j)\times f_k(\tuple x_k)$ is analogous; just add $\bn$ in front of the existential quantification.
\item Let $\theta=f_i(\tuple x_i)=\SUM_{\tuple x_k}f_j(\tuple x_k\tuple x_j)$. Then define
\[
\Theta\coloneqq\exists\alpha\beta((\alpha=0\leftrightarrow\tuple x_i=\tuple y_i)\wedge(\beta=0\leftrightarrow\tuple x_j=\tuple y_j)\wedge\tuple x\alpha\approx\tuple x\beta).
\]
The negated case $\neg f_i(\tuple x_i)=\SUM_{\tuple x_k}f_j(\tuple x_k\tuple x_j)$ is analogous; just add $\bn$ in front of the existential quantification. The proof is similar to the previous one, so it is omitted.
\item If $\theta=\theta_0\wedge\theta_1$, then $\Theta=\Theta_0\wedge\Theta_1$. The claim directly follows from semantics of conjunction. 
\item Let $\theta=\theta_0\lor\theta_1$. Then define
\[
\Theta\coloneqq\exists z(\pci{\tuple x}{z}{z}\wedge((\Theta_0\wedge z=0)\lor(\Theta_1\wedge \neg z=0))).
\]
Assume first that $(\mathcal{A},p,f_1,\dots,f_n)\models\theta(\tuple a)$ for all $\tuple a\in M$. Then there are $M_0,M_1$ such that $M_0\cup M_1=M$, $M_0\cap M_1=\emptyset$, and $(\mathcal{A},p,f_1,\dots,f_n)\models\theta_i(\tuple a)$  for all $\tuple a\in M_i$. Define $F\colon Y\to p_A$ such that $F(s)=c_i$ when $s(\tuple x)\in M_i$, where $c_i$ is the distribution defined as
\begin{align*}
c_i(a)\coloneqq
\begin{cases}
1 \text{ if } a=i,\\
0 \text{ otherwise}.
\end{cases}
\end{align*}
Let $\mathbb{Z}_i\coloneqq\X(M_i/\tuple x)(f_1/\tuple y_1)\dots(f_n/\tuple y_n)(c_i/z)$ and $k=|M_0|/|M|$. Now $\mathbb{Z}=\Y(F/z)=\mathbb{Z}_0\sqcup_k\mathbb{Z}_1$, and we have $\mathcal{A}\models_{\mathbb{Z}}\pci{\tuple x}{z}{z}$, $\mathcal{A}\models_{\mathbb{Z}_0}\Theta_0\wedge z=0$, and $\mathcal{A}\models_{\mathbb{Z}_1}\Theta_1\wedge \neg z=0$. By locality, this implies that $\mathcal{A}\models_Y\Theta$. 

Assume then that $\mathcal{A}\models_\Y\Theta$. 
Let $F\colon Y\to p_A$ be such that $\mathcal{A}\models_{\mathbb{Z}}\pci{\tuple x}{z}{z}\wedge((\Theta_0\wedge z=0))\lor(\Theta_1\wedge \neg z=0)$ for $\mathbb{Z}=\Y(F/z)$. 
Let then $k\mathbb{Z}_0'=\mathbb{Z}_{z=0}$ and $(1-k)\mathbb{Z}_1'=\mathbb{Z}_{z=1}$ for $k=|\mathbb{Z}_{z=0}|$. 
Now, we also have $\mathcal{A}\models_{\mathbb{Z}_i'}\Theta_i$ for $i=0,1$. 
Since $\mathcal{A}\models_{\mathbb{Z}}\pci{\tuple x}{z}{z}$, we have either $\mathbb{Z}_{\tuple x=\tuple a}=\mathbb{Z}_{\tuple xz=\tuple a0}$ or $\mathbb{Z}_{\tuple x=\tuple a}=\mathbb{Z}_{\tuple xz=\tuple a1}$ for all $a\in M$. 
We get that $\mathbb{Z}_{\z=0}=\mathbb{Z}_{\tuple x\in M_0}$ for some $M_0\subseteq M$. 
Thus, $\mathbb{Z}_0'=\frac{|M|}{|M_0|}(\X(M/\tuple x)(f_1/\tuple y_1)\dots(f_n/\tuple y_n))_{\tuple x\in M_0}=\X(M_0/\tuple x)(f_1/\tuple y_1)\dots(f_n/\tuple y_n)$. 
Hence, $(\mathcal{A},p,f_1,\dots,f_n)\models\theta_0(\tuple a)$ for all $\tuple a\in M_0$. 
We obtain $(\mathcal{A},p,f_1,\dots,f_n)\models\theta_1(\tuple a)$ for all $\tuple a\in M\setminus M_0$ by an analogous argument. 
As a result, we get that $(\mathcal{A},p,f_1,\dots,f_n)\models\theta(\tuple a)$ for all $\tuple a\in M$.
\end{enumerate}
\end{proof}

\section{Probabilistic logics and entropy atoms}
In this section we consider extending probabilistic team semantics with novel entropy atoms.
For a discrete random variable $X$, with possible outcomes $x_1, ..., x_n$ occuring with probabilities $\mathrm{P}(x_1), ..., \mathrm{P}(x_n)$, the Shannon entropy of $X$ is given as:
\[
\entropy{X} \dfn -\sum_{i=1}^n {\mathrm{P}(x_i) \log \mathrm{P}(x_i)},
\]
The base of the logarithm does not play a role in this definition (usually it is assumed to be $2$).
For a set of discrete random variables, the entropy is defined in terms of the vector-valued random variable it defines. Given three sets of discrete random variables $X,Y,Z$, it is known that $X$ is conditionally independent of $Y$ given $Z$ (written $X \perp\!\!\!\perp Y \mid Z$)  if and only if the conditional mutual information $\mathrm{I}(X;Y| Z)$ vanishes. 
Similarly, functional dependence of $Y$ from $X$ holds if and only if the conditional entropy $H(Y|X)$ of $Y$ given $X$ vanishes. Writing $UV$ for the union of two sets $U$ and $V$, we note that $\mathrm{I}(X;Y| Z)$ and $H(Y|X)$ can respectively be expressed as $H(ZX) + H(ZY) - H(Z) - H(ZXY)$ and $H(XY) - H(X)$.
Thus many familiar dependency concepts over random variables translate into linear equations over Shannon entropies. 
In what follows, we shortly consider similar information-theoretic approach to dependence and independence in probabilistic team semantics.

Let $\X\colon X \to [0,1]$ be a probabilistic team over a finite structure $\A$ with universe $A$. Let $\tuple x$ be a $k$-ary sequence of variables from the domain of $\X$.
Let $P_{\tuple x}$ be the vector-valued random variable, where $P_{\tuple x}(\tuple a)$ is the probability that $\tuple x$ takes value $\tuple a $ in the probabilistic team $\X$.
The \emph{Shannon entropy} of 
$\tuple x$ in $\X$ is defined as follows:
\begin{equation}\label{eq:entropy}
\entropyteam{\tuple x}{\X} \dfn -\sum_{\tuple a \in A^k} \mathrm{P}_{\tuple x}(\tuple a) \log \mathrm{P}_{\tuple x}(\tuple a).
\end{equation}
Using this definition we now define the concept of an entropy atom.
\begin{definition}[Entropy atom]
Let $\tuple x$ and $\tuple y$ be two sequences of variables from the domain of $\X$. These sequences may be of different lengths. The \emph{entropy atom} is an expression of the form $\entropyatom{\tuple x}{\tuple y}$, and it is given the following semantics:
\[
\A\models_\X \entropyatom{\tuple x}{\tuple y} \iff \entropyteam{\tuple x}{\X} = \entropyteam{\tuple y}{\X}.
\]	
\end{definition}

We then define \emph{entropy logic} $\FO(\eH)$ as the logic obtained by extending first-order logic with entropy atoms. The entropy atom is relatively powerful compared to our earlier atoms, since, as we will show next, it encapsulates many familiar dependency notions such as dependence and conditional independence. 
\begin{theorem}
The following equivalences hold over probabilistic teams of finite structures with two distinct constants $0$ and $1$:
\begin{enumerate}
\item $\dep(\tuple x, \tuple y)\equiv \entropyatom{\tuple x}{\tuple x \tuple y}$.
\item $\pmi{\tuple x}{\tuple y} \equiv \phi$, where $\phi$ is defined as
\begin{align*}
\forall z \exists \tuple u \tuple v  \Big( &\big[z=0 \to \big(\dep(\tuple u, \tuple x) \land \dep(\tuple x, \tuple u) \land \dep(\tuple v, \tuple x\tuple y) \land \dep(\tuple x\tuple y, \tuple v)\big)\big]  \land\\
& \big[z=1 \to \big(\dep(\tuple u, \tuple y) \land \dep(\tuple y, \tuple u)   \land \tuple v =\tuple 0\big)\big] \land \\
& \big[(z=0 \lor z=1 ) \to \entropyatom{\tuple uz}{\tuple vz}\big]\Big),
\end{align*}
where $|\tuple u|=\max\{|\tuple x|,\tuple y|\}$ and $|\tuple v | = |\tuple x \tuple y|$.
\end{enumerate}
\end{theorem}
\begin{proof}
The translation of the dependence atom simply expresses that the conditional entropy of $\vec{y}$ given $\vec{x}$ vanishes, which expresses that $\vec{y}$ depends functionally on $\vec{x}$.
 
 Consider the translation of the independence atom. Observe that $\phi$ essentially restricts attention to that subteam $\Y$ in which the universally quantified variable $z$ is either $0$ or $1$. 
There, the weight distribution of $\tuple uz$ is obtained by vertically stacking together halved weight distributions of $\tuple x$ and $\tuple y$. Similarly,  $\tuple vz$ corresponds to halving and vertical stacking of $\tuple x \tuple y$ and a dummy constant distribution $\tuple 0$. Consider now the effect of halving the weights of the entropy function given in \eqref{eq:entropy}:
\begin{align*}
\entropy{\frac{1}{2}X} & =  -\sum_{i=1}^n {\frac{1}{2}\mathrm{P}(x_i) \log \frac{1}{2}\mathrm{P}(x_i)}\\
&= -\frac{1}{2}\sum_{i=1}^n {\mathrm{P}(x_i) (\log \frac{1}{2} + \log \mathrm{P}(x_i))}\\
& =  -\frac{1}{2}\sum_{i=1}^n \mathrm{P}(x_i) \log \frac{1}{2})   -\frac{1}{2}\sum_{i=1}^n \mathrm{P}(x_i) \log \mathrm{P}(x_i)\\
& = \frac{1}{2} + \frac{1}{2}\entropy{X}.
\end{align*}
Let us turn back to our subteam $\Y$, obtained by quantification and split disjunction from some initial team $\X$. This subteam has to satisfy $\entropyatom{\tuple uz}{\tuple vz}$. What this amounts to, is the following
\begin{align*}
\entropyteam{\tuple uz}{\Y} = \entropyteam{\tuple yz}{\Y} &\iff \entropyteam{\frac{1}{2}\tuple x}{\X} + \entropyteam{\frac{1}{2}\tuple y}{\X} = \entropyteam{\frac{1}{2}\tuple xy}{\X} + \entropyteam{\frac{1}{2}\tuple 0}{\X}\\
& \iff 1 + \frac{1}{2}\entropyteam{\tuple x}{\X} + \frac{1}{2}\entropyteam{\tuple y}{\X} = \frac{1}{2} + \frac{1}{2}\entropyteam{\tuple x\tuple y}{\X} + \frac{1}{2}\entropyteam{\tuple 0}{\X}\\
& \iff \entropyteam{\tuple x}{\X} + \entropyteam{\tuple y}{\X} = \entropyteam{\tuple x \tuple y}{\X}.
\end{align*}
Thus, the translation captures the entropy condition of the independence atom. 
\end{proof}

Since conditional independence can be expressed with marginal independence, i.e.,  $\FO(\cpind)\equiv\FO(\pind)$ \cite[Theorem 11]{HannulaHKKV19}, we obtain the following corollary:

\begin{corollary}\label{cor:FO-ind-eh}
$\FO(\cpind)\leq \FO(\eH)$.
\end{corollary}

It is easy to see at this point that entropy logic and its extension with negation are subsumed by second-order logic over the reals with exponentiation. 
\begin{theorem} \label{thm:H-ESO-H-neg-SO}
$\FO(\eH)\leq \ESO_{\mathbb{R}}(+,\times, \log)$ and  $\FO(\eH,\sim)\leq \SO_{\mathbb{R}}(+,\times, \log)$. 
\end{theorem}
\begin{proof}
The translation is similar to the one in Theorem \ref{thm_ind<SO}, so it suffices to notice that the entropy atom $\eH(\tuple x)=\eH(\tuple y)$ can be expressed as
\[
\SUM_{\tuple x}(\SUM_{\tuple z}f(\tuple x,\tuple z)\log \SUM_{\tuple z}f(\tuple x,\tuple z))=\SUM_{\tuple y}(\SUM_{\tuple z'}f(\tuple y,\tuple z')\log \SUM_{\tuple z'}f(\tuple y,\tuple z')).
\]
Since $\SUM$ can be expressed in $\ESO_{\mathbb{R}}(+,\times, \log)$ and $\SO_{\mathbb{R}}(+,\times, \log)$, we are done.
\end{proof}




\section{Logic for first-order probabilistic dependecies
}


Here, we define the logic $\FOPT$, which was introduced in \cite{HHK22}.\footnote{In \cite{HHK22}, two sublogics of $\FOPT$, called $\mathrm{FOPT}(\leq^\delta)$ and $\mathrm{FOPT}(\leq^{\delta},\cpind^{\delta})$, were also considered. Note that the results of this section also hold for these sublogics.} Let $\delta$ be a quantifier- and disjunction-free first-order formula, i.e., $\delta\Coloneqq\lambda \mid \neg\delta \mid (\delta\wedge\delta)$ for a first-order atomic formula $\lambda$ of the vocabulary $\tau$. Let $x$ be a first-order variable. The syntax for the logic $\FOPT$ over a vocabulary $\tau$ is defined as follows:
\[
\phi\Coloneqq\delta\mid (\delta|\delta)\leq(\delta|\delta) \mid\wcn\phi \mid (\phi\wedge\phi) \mid (\phi\vvee\phi) \mid \existso x\phi \mid \forallo x\phi.
\]

Let $\X\colon X\to\mathbb{R}_{\geq 0}$ be any probabilistic team, not necessarily a probability distribution. The semantics for the logic is defined as follows:
\begin{itemize}
\item[] $\A\models_{\X}\delta$ iff $\A\models_{s}\delta$ for all $s\in \supp(\X)$.
\item[] $\A\models_{\X}(\delta_0|\delta_1)\leq(\delta_2|\delta_3)$ iff $|\X_{\delta_0\wedge\delta_1}|\cdot|\X_{\delta_3}|\leq|\X_{\delta_2\wedge\delta_3}|\cdot|\X_{\delta_1}|$.
\item[] $\A\models_{\X}\wcn\phi$ iff $\A\not\models_{\X}\phi$ or $\X$ is empty.
\item[] $\A\models_{\X}\phi\wedge\psi$ iff $\A\models_{\X}\phi$ and $\A\models_{\X}\psi$.
\item[] $\A\models_{\X}\phi\vvee\psi$ iff $\A\models_{\X}\phi$ or $\A\models_{\X}\psi$.
\item[] $\A\models_{\X}\existso x\phi$ iff $\A\models_{\X(a/x)}\phi$ for some $a\in A$. 
\item[] $\A\models_{\X}\forallo x\phi$ iff $\A\models_{\X(a/x)}\phi$ for all $a\in A$.
\end{itemize}

Next, we present some useful properties of $\FOPT$. 

\begin{proposition}[\textbf{Locality}, {\cite[Prop.~3.2]{HHK22}}]\label{locality}
Let $\phi$ be any ${\FOPT}[\tau]$-formula. Then for any set of variables $V$, any $\tau$-structure $\A$, and any probabilistic team $\X\colon X\to\mathbb{R}_{\geq 0}$ such that $\Fr(\phi)\subseteq V\subseteq D$, 
\[
\A\models_{\X}\phi \iff \A\models_{\X\restriction{V}}\phi.
\]
\end{proposition}

Over singleton traces the expressivity of $\FOPT$ coincides with that of $\FO$. For $\phi\in\FOPT$, let $\phi^*$ denote the $\FO$-formula obtained by replacing the symbols $\wcn, \vvee, \existso$, and $\forallo$ by $\neg,\lor,\exists$, and $\forall$, respectively, and expressions of the form $(\delta_0\mid\delta_1)\leq(\delta_2\mid\delta_3)$ by the formula $\neg \delta_0\lor \neg\delta_1\lor\delta_2\lor\neg\delta_3$.
\begin{proposition}[Singleton equivalence]\label{prop:singleton}
Let $\phi$ be a ${\FOPT}[\tau]$-formula, $\A$ a $\tau$ structure, and $\X$ a probabilistic team of $\A$ with support $\{s\}$. Then $\A\models_{\X}\phi$ iff $\A\models_{s}\phi^*$. 
\end{proposition}
\begin{proof}
The proof proceeds by induction on the structure of formulas. The cases for literals and Boolean connectives are trivial. The cases for quantifiers are immediate once one notices that interpreting the quantifiers $\existso$ and $\forallo$ maintain singleton supportness. We show the case for $\leq$. Let $\lVert \delta \rVert_{\A,s}=1$ if $\A\models_s\delta$, and $\lVert \delta \rVert_{\A,s}=0$ otherwise. Then
\begin{align*}
\A\models_{\X}(\delta_0\mid\delta_1)\leq(\delta_2\mid\delta_3)&\iff |\X_{\delta_0\wedge\delta_1}|\cdot|\X_{\delta_3}|\leq|\X_{\delta_2\wedge\delta_3}|\cdot|\X_{\delta_1}|\\
&\iff \lVert \delta_0\wedge\delta_1 \rVert_{\A,s} \cdot \lVert \delta_3 \rVert_{\A,s} 
\leq \lVert \delta_2\wedge\delta_3 \rVert_{\A,s} \cdot \lVert\delta_1 \rVert_{\A,s}\\
&\iff \A\models_s \neg \delta_0\lor \neg\delta_1\lor\delta_2\lor\neg\delta_3.
\end{align*}
The first equivalence follows from the semantics of $\leq$ and the second follows from the induction hypotheses after observing that the support of $\X$ is $\{s\}$. The last equivalence follows via a simple arithmetic observation.
\end{proof}



The following theorem follows directly from Propositions \ref{locality} and \ref{prop:singleton}.
\begin{theorem}\label{thm:FOPT}
For sentences we have that $\FOPT\equiv \FO$.	
\end{theorem}

For a logic $L$, we write $\MC(L)$ for the following variant of the model checking problem: given a \textit{sentence} $\phi\in L$ and a structure $\A$, decide whether $\A\models\phi$. 
The above result immediately yields the following corollary.

\begin{corollary}\label{cor:sent-mcfopt-pspace}
$\MC(\FOPT)$ is $\PSPACE$-complete.	
\end{corollary}
\begin{proof} This follows directly from the linear translation of $\FOPT$-sentences into   equivalent   $\FO$ -sentences  of Theorem \ref{thm:FOPT} and the well-known fact that the model-checking problem of $\FO$ is $\PSPACE$-complete.
\end{proof}

\begin{theorem}\label{thm:openforms-fopt-fpind-neg}
$\FOPT\le \FO(\cpind,\sim)$ and $\FOPT$ is non-comparable to $\FO(\cpind)$ for open formulas.
\end{theorem}
\begin{proof} 

We begin the proof of the first claim by showing that $\FOPT\leq\ESOr$. Note that we may use numerical terms of the form $i\leq j$ in $\ESOr$, because they can be expressed by the formula $\exists f\exists g(g\times g=f\wedge i+f=j)$.

Let formula $\phi(\tuple v)\in\FOPT$ be such that its free-variables are from $\tuple v=(v_1,\dots,v_k)$. Then there is a formula $\psi_{\phi}(f)\in\ESOr$ with exactly one free function variable such that for all structures $\A$ and all probabilistic teams $\X\colon X\to\mathbb{R}_{\geq 0}$, $\A\models_\X\phi(\tuple v)$ if and only if $(\A,f_\X)\models\psi_\phi(f)$, where $f_\X\colon A^k\to\mathbb{R}_{\geq 0}$ is a function such that $f_\X(s(\tuple v))=\X(s)$ for all $s\in X$.

We may assume that 
the formula is in the form $\phi=Q_1^1x_1\dots Q_n^1x_n\theta(\tuple v,\tuple x)$, where $Q_i\in\{\exists,\forall\}$ and $\theta$ is quantifier-free. We begin by defining inductively a formula $\theta^*(f,\tuple x)$ for the subformula $\theta(\tuple v,\tuple x)$. Note that in the following $\chi_\delta$ refers to the characteristic function of $\delta$, i.e., $\chi_\delta\colon A^{k+n}\to\{0,1\}$ such that $\chi_\delta(\tuple a)=1$ if and only if $\A\models\delta(\tuple a)$. For simplicity, we only write $\theta^*(f,\tuple x)$ despite the fact that $\theta^*$ may contain free function variables $\chi_\delta$ in addition to the variables $f,\tuple x$.

\begin{enumerate}[1.]
\item If $\theta(\tuple v,\tuple x)=\delta(\tuple v,\tuple x)$, then $\theta^*(f,\tuple x)\coloneqq\forall\tuple v(f(\tuple v)=0\lor\chi_\delta(\tuple v,\tuple x)=1)$.
\item If $\theta(\tuple v,\tuple x)=(\delta_0\mid\delta_1)\leq(\delta_2\mid\delta_3)(\tuple v,\tuple x)$, then
\begin{align*}
\theta^*(f,\tuple x)\coloneqq&\SUM_{\tuple v}(f(\tuple v)\times\chi_{\delta_0\wedge\delta_1}(\tuple v,\tuple x))\times\SUM_{\tuple v}(f(\tuple v)\times\chi_{\delta_3}(\tuple v,\tuple x))\\
&\leq\SUM_{\tuple v}(f(\tuple v)\times\chi_{\delta_2\wedge\delta_3}(\tuple v,\tuple x))\times\SUM_{\tuple v}(f(\tuple v)\times\chi_{\delta_1}(\tuple v,\tuple x)).
\end{align*}
\item If $\theta(\tuple v,\tuple x)=\wcn\theta_0(\tuple v,\tuple x)$, then $\theta^*(f,\tuple x)\coloneqq\theta_0^{*\neg}(f,\tuple x)\lor\forall\tuple vf(\tuple v)=0$, where $\theta_0^{*\neg}$ is obtained from $\neg\theta_0^*$ by pushing the negation in front of atomic formulas.
\item If $\theta(\tuple v,\tuple x)=(\theta_0\circ\theta_1)(\tuple v,\tuple x)$, where $\circ\in\{\wedge,\vvee\}$, then $\theta^*(f,\tuple x)\coloneqq(\theta_0^*\star\theta_1^*)(\tuple x)$, where $\star\in\{\wedge,\lor\}$, respectively.
\end{enumerate}

For each $\delta$, we define a formula $\xi_{\delta}\in\ESOr$, which says that $\chi_\delta$ is the characteristic function of $\delta$. Let $\tuple y=(y_1,\dots,y_{k+n})$ and define $\xi_\delta$ as follows:
\begin{enumerate}[1.]
\item If $\delta(\tuple y)=R(y_{i_1},\dots,y_{i_l})$, where $1\leq i_1,\dots, i_l\leq k+n$, then $\xi_\delta\coloneqq\forall\tuple y((\chi_\delta(\tuple y)=1\leftrightarrow R(y_{i_1},\dots,y_{i_l}))\wedge(\chi_\delta(\tuple y)=0\leftrightarrow \neg R(y_{i_1},\dots,y_{i_l}))$.
\item If $\delta(\tuple y)=\neg\delta_0(\tuple y)$, then $\xi_\delta\coloneqq\forall\tuple y(\chi_{\delta_0}(\tuple y)+\chi_{\neg\delta_0}(\tuple y)=1)$.
\item If $\delta(\tuple y)=(\delta_0\wedge\delta_1)(\tuple y)$, then $\xi_\delta\coloneqq\forall\tuple y(\chi_{\delta_0\wedge\delta_0}(\tuple y)=\chi_{\delta_0}(\tuple y)\times\chi_{\delta_1}(\tuple y))$
\end{enumerate}

Let $\delta_1,\dots,\delta_m$ be a list such that each $\delta_i$, $1\leq i\leq m$, is a subformula of some formula $\delta$ that appears in a function symbol $\chi_{\delta}$ of the formula $\theta^*(f,\tuple x)$. Now, we can define
\[
\psi_\phi(f)\coloneqq\bigexists_{1\leq i\leq m}\chi_{\delta_i}\left(Q_1 x_1\dots Q_k x_k\theta^*(f,\tuple x)\wedge\bigwedge_{1\leq i\leq m}\xi_{\delta_i}(\chi_{\delta_1},\dots,\chi_{\delta_m})\right).
\]
This shows that $\FOPT\leq\ESOr$.  The first claim now follows, since $\ESOr\leq\SO_{\mathbb{R}}(+,\times)\equiv \FO(\cpind,\sim)$.

We will prove the second claim now. In the proof of Proposition \ref{prop:FO-ind-strict-FO-ind-sim}, it was noted that the formula $\sim x \pind_y z$ cannot be expressed in $\FO(\cpind)$. This is not the case for $\FOPT$ as it contains the Boolean negation, and thus the formula $\sim x \pind_y z$ can be expressed in  $\FOPT$ by the results of Section 4.2 in \cite{HHK22}.

On the other hand, we have $\FO(\dep(\dots))\leq\FO(\cpind)$ (Prop. \ref{prop_depcpind}). Since on the level of sentences, $\FO(\dep(\dots))$ is equivalent to existential second-order logic \cite{vaananen07}, there is a sentence $\phi\in\FO(\cpind)$ such that for all $\X\colon X\to[0,1]$, $\A\models_\X\phi$ iff a undirected graph $\A=(V,E)$ is 2-colourable. Since over singleton traces the expressivity of $\FOPT$ coincides with $\FO$, the sentence $\phi$ cannot be expressed in $\FOPT$, as 2-colourability cannot be expressed in $\FO$.

\end{proof}

\section{Complexity of satisfiability, validity and model checking
 }

We now define satisfiability and validity in the context of probabilistic team semantics. Let $\phi\in\FO(\cpind,\sim,\approx)$. The formula $\phi$ \textit{is satisfiable in a} \textit{structure} $\A$ if $\A\models_{\X}\phi$ for some probabilistic team $\X$, and $\phi$ \textit{is valid in a structure} $\A$ if $\A\models_{\X}\phi$ for all probabilistic teams $\X$ over $\Fr(\phi)$. The formula $\phi$ is \textit{satisfiable} if there is a structure $\A$ such that $\phi$ is satisfiable in $\A$, and $\phi$ is \textit{valid} if $\phi$ is valid in $\A$ for all structures $\A$.

For a logic $L$, the satisfiability problem  $\SAT(L)$ and the validity problem $\VAL(L)$ are defined as follows: given a formula $\phi\in L$, decide whether $\phi$ is satisfiable (or valid, respectively). For the model checking problem $\MC(L)$, we consider the following variant: given a \textit{sentence} $\phi\in L$ and a structure $\A$, decide whether $\A\models\phi$.


\begin{theorem}\label{thm:MC-inc}
	$\MC(\probinclogic)$ is in $\EXPTIME$ and $\PSPACE$-hard.
\end{theorem}
\begin{proof}
First note that $\probinclogic$ is clearly a conservative extension of $\FO$, as it is easy to check that probabilistic semantics and Tarski semantics agree on first-order formulas over singleton traces. The hardness now follows from this and the fact that model checking problem for $\FO$ is $\PSPACE$-complete.

For upper bound, notice first that any $\probinclogic$-formula $\phi$ can be reduced to an almost conjunctive formula $\psi^*$ of $\ESO_R(+,\leq,\SUM)$~\cite[Lem,~17]{HannulaV22}.
Then the desired bounds follow due to the reduction from Proposition~3 in~\cite{HannulaV22}.
The mentioned reduction yields families of systems of linear inequalities $\mathcal S$ from a structure $\A$ and assignment $s$ such that a system $S\in\mathcal S$ has a solution if and only if $\A\models_s\phi$.
For a $\probinclogic$-formula $\phi$, this transition requires exponential time and this yields membership in $\EXPTIME$.
	
	
\end{proof}

We now prove the following lemma, which will be used to prove the upper-bounds in the next three theorems.
\begin{lemma}\label{probind_lemma}
Let $\A$ be a finite structure and $\phi\in\FO(\cpind,\sim)$. Then there is a first-order sentence $\psi_{\phi,\A}$ over vocabulary $\{+,\times,\leq,0,1\}$ such that $\phi$ is satisfiable in $\A$ if and only if $(\mathbb{R},+,\times,\leq,0,1)\models \psi_{\phi,\A}$.
\end{lemma}
\begin{proof} Let $\phi$ be such that its free variables are from $\tuple v=(v_1,\dots ,v_k)$. By locality (Prop. \ref{locality2}), we may restrict to the teams over the variables $\{v_1,\dots ,v_k\}$. Define a fresh first-order variable $s_{\tuple v=\tuple a}$ for each $\tuple a\in A^k$. The idea is that the variable $s_{\tuple v=\tuple a}$ represents the weight of the assignment $s$ for which $s(\tuple v)=\tuple a$. For notational simplicity, assume that $A=\{1,\dots,n\}$. Thus, we can write $\tuple s=(s_{\tuple v=\tuple 1},\dots  ,s_{\tuple v=\tuple n})$ for the tuple that contains the variables for all the possible assignments over $\tuple v$. Define then
\[
\psi_{\phi,\A}\coloneqq\exists s_{\tuple v=\tuple 1}\dots  s_{\tuple v=\tuple n}\left(\bigwedge_{\tuple a} 0\leq s_{\tuple v=\tuple a}\wedge \neg 0=\sum_{\tuple a} s_{\tuple v=\tuple a}\wedge\phi^*(\tuple s)\right) ,
\]
where $\phi^*(\tuple s)$ is constructed as follows:
\begin{itemize}
\item[$\bullet$] If $\phi(\tuple v)=R(v_{i_1},\dots,v_{i_l})$ or $\phi(\tuple v)=\neg R(v_{i_1},\dots,v_{i_l})$ where $1\leq i_1,\dots,i_l\leq k$, then $\phi^*(\tuple s)\coloneqq\bigwedge_{s\not\models\phi}s=0$.
\item[$\bullet$] If $\phi(\tuple v)=\pci{\tuple v_0}{\tuple v_1}{\tuple v_2}$ for some $\tuple v_3$ such that $\tuple v={\tuple v_0}{\tuple v_1}{\tuple v_2}\tuple v_3$, then 
\begin{align*}
\phi^*(\tuple s)\coloneqq\bigwedge_{\tuple a_0\tuple a_1\tuple a_2}\Big(
&\sum_{\tuple b_2\tuple b_3} s_{\tuple v=\tuple a_0\tuple a_1\tuple b_2\tuple b_3}\times \sum_{\tuple b_1\tuple b_3} s_{\tuple v=\tuple a_0\tuple b_1\tuple a_2\tuple b_3}=\\
& \sum_{\tuple b_3} s_{\tuple v=\tuple a_0\tuple a_1\tuple a_2\tuple b_3} \times \sum_{\tuple b_1\tuple b_2\tuple b_3} s_{\tuple v=\tuple a_0\tuple b_1\tuple b_2\tuple b_3}
\Big),
\end{align*}
\item[$\bullet$] If $\phi(\tuple v)=\bn\theta_0(\tuple v)$ or $\phi(\tuple v)=\theta_0(\tuple v)\wedge\theta_1(\tuple v)$, then $\phi^*(\tuple s)\coloneqq\neg\theta_0^*(\tuple s)$ or $\phi^*(\tuple s)\coloneqq\theta_0^*(\tuple s)\wedge\theta_1^*(\tuple s)$, respectively.
\item[$\bullet$] If $\phi(\tuple v)=\theta_0(\tuple v)\lor\theta_1(\tuple v)$, then 
\begin{align*}
\phi^*(\tuple s)\coloneqq\exists t_{\tuple v=\tuple 1}r_{\tuple v=\tuple 1} \dots t_{\tuple v=\tuple n}r_{\tuple v=\tuple n} \biggl(&\bigwedge_{\tuple a}( 0\leq t_{\tuple v=\tuple a}\wedge0\leq r_{\tuple v=\tuple a}\wedge\\
& s_{\tuple v=\tuple a}=t_{\tuple v=\tuple a}+r_{\tuple v=\tuple a})\wedge\theta_0^*(\tuple t)\wedge\theta_1^*(\tuple r)\biggr).
\end{align*}
\item[$\bullet$] If $\phi(\tuple v)=\exists x\theta_0(\tuple v,x)$, then 
\begin{align*}
\phi^*(\tuple s)\coloneqq\exists t_{\tuple vx=\tuple 11} \dots t_{\tuple vx=\tuple nn}\biggl(\bigwedge_{\tuple ab}( 0\leq t_{\tuple vx=\tuple ab}\wedge s_{\tuple v=\tuple a}=\sum_{c=1}^nt_{\tuple vx=\tuple ac})\wedge\theta_0^*(\tuple t)\biggr).
\end{align*}
\item[$\bullet$] If $\phi(\tuple v)=\forall x\theta_0(\tuple v,x)$, then 
\begin{align*}
\phi^*(\tuple s)\coloneqq\exists t_{\tuple vx=\tuple 11} \dots t_{\tuple vx=\tuple nn}\biggl(&\bigwedge_{\tuple ab}( 0\leq t_{\tuple vx=\tuple ab}\wedge s_{\tuple v=\tuple a}=\sum_{c=1}^nt_{\tuple vx=\tuple ac}\wedge\\
&\bigwedge_{cd}t_{\tuple vx=\tuple ac}=t_{\tuple vx=\tuple ad})\wedge\theta_0^*(\tuple t)\biggr).
\end{align*}
\end{itemize}
\end{proof}

\begin{theorem}\label{thm:MCFO-ind}
	$\MC(\FO(\cpind))$ is in $\EXPSPACE$ and $\NEXPTIME$-hard.
\end{theorem}

\begin{proof}
For the lower bound, we use the fact that dependence atoms can be expressed by using probabilistic independence atoms.
Let $\A$ be a structure and $\X$ be a probabilistic team over $\A$. 
Then $\A\models_\X \dep(\tuple x ,\y) \iff  \A\models_\X \y\pind_\x \y$~\cite[Prop.~3]{HannulaHKKV19}.
The $\NEXPTIME$-hardness follows since the model checking problem for $\FO(\dep(\dots))$ is $\NEXPTIME$-complete~\cite[Thm.~5.2]{Gradel13}. 


The upper-bound follows from the fact that when restricted to $\FO(\cpind)$, the exponential translation in Lemma \ref{probind_lemma} is an existential sentence, and the existential theory of the reals is in $\PSPACE$. 
\end{proof}

\begin{theorem}\label{thm:sent-mc-sim-ind-expspace-aexp-poly}
$\MC(\FO(\sim,\cpind))$ is in $\threeEXPSPACE$ and $\AEXPTIME[\poly]$-hard.
\end{theorem} 
\begin{proof} 
	We first prove the lower bound through a reduction from the satisfiability problem for propositional team-based logic, that is, $\SAT(\PL(\sim))$.
	Given a $\PL(\sim)$-formula $\phi$, the problems asks whether there is a team $T$ such that $T\models\phi$?
	Let $\phi$ be a $\PL(\sim)$-formula over propositional variables $p_1,\ldots, p_n$.
	For $i\leq n$, let $x_i$ denote a variable corresponding to the proposition $p_i$.
	Let $\A=\{0,1\}$ be the structure over empty vocabulary.
	Then, $\phi$ is satisfiable iff $\exists p_1\ldots \exists p_n \phi$ is satisfiable iff $\A \models_{\{\emptyset\}} \exists x_1\ldots \exists x_n \phi'$, where $\phi'$ is a $\FO(\sim)$-formula obtained from $\phi$ by simply replacing each proposition $p_i$ by the variable $x_i$.
	This gives $\AEXPTIME$-hardness of $\MC(\FO(\sim))$ (and consequently, of $\MC(\FO(\sim,\cpind))$) since the satisfiability for $\PL(\sim)$ is $\AEXPTIME$-complete~\cite{HannulaKVV18}.

The upper-bound follows from the exponential translation from $\FO(\sim,\cpind)$ to real arithmetic in Lemma \ref{probind_lemma} and the fact that the full theory of the reals is in $\twoEXPSPACE$.
\end{proof}

\begin{theorem} \label{thm:ind-sim-val-sat}
	$\SAT(\FO(\cpind,\sim))$ is $\RE$- and $\VAL(\FO(\cpind,\sim))$ is $\coRE$-complete.
\end{theorem}
\begin{proof}

It suffices to prove the claim for $\SAT(\FO(\cpind,\sim))$, since the claim for $\VAL(\FO(\cpind,\sim))$ follows from the fact that $\FO(\cpind,\sim)$ has the Boolean negation.

For the lower bound, note that $\FO(\cpind,\sim)$ is a conservative extension of $\FO$, and hence the claim follows from the r.e.-hardness of $\textsc{SAT}(\logicFont{FO})$ over the finite.

For the upper-bound, we use Lemma \ref{probind_lemma}. Let $\phi$ be a satisfiable formula of $\FO(\cpind,\sim)$. We can verify that $\phi\in\textsc{SAT}(\FO(\cpind,\sim))$ by going through all finite structures until we come across a structure in which $\phi$ is satisfiable. Hence, it suffices to show that for any finite structure $\A$, it is decidable to check whether $\phi$ is satisfiable in $\A$.  For this, construct a sentence $\psi_{\A,\phi}$ as in Lemma \ref{probind_lemma}. Then $\psi_{\A,\phi}$ is such that $\phi$ is satisfiable in $\A$ iff $(\mathbb{R},+,\times,\leq,0,1)\models\psi_{\A,\phi}$. Since real arithmetic is decidable, we now have that $\SAT(\FO(\cpind,\sim))$ is $\RE$-complete. 
\end{proof}

\begin{corollary}
$\SAT(\FO(\approx))$ and $\SAT(\FO(\cpind))$ are $\RE$- and $\VAL(\FO(\approx))$ and $\VAL(\FO(\cpind))$ are $\coRE$-complete.
\end{corollary}
\begin{proof}
The lower bound follows from the fact that $\FO(\approx)$ and $\FO(\cpind)$ are both conservative extensions of $\FO$. We obtain the upper bound from the previous theorem, since $\FO(\cpind,\sim)$ includes both $\FO(\approx)$ and $\FO(\cpind)$.
\end{proof} 




\section{Conclusion}
We have studied the expressivity and  complexity of various logics in probabilistic team semantics with the Boolean negation. 
Our results give a quite comprehensive picture of the relative expressivity of these logics and their relations to numerical variants of (existential) second-order logic. 
An interesting question for further study is to determine the exact complexities of the decision problems studied in Section 8. 
Furthermore, dependence atoms based on various notions of entropy deserve further study, as do the connections of probabilistic team semantics to the field of information theory.

\subsubsection{Acknowledgements} 
The first author is supported by the ERC grant 101020762.
The second author is supported by Academy of Finland grant 345634.
The third author is supported by Academy of Finland grants 338259 and 345634. 
The fifth author appreciates funding by the German Research Foundation (DFG), project ME 4279/3-1. 
The sixth author is partially funded by the German Research Foundation (DFG), project VI 1045/1-1.

\bibliographystyle{splncs04}
\bibliography{main}

\begin{thebibliography}{10}
\providecommand{\url}[1]{\texttt{#1}}
\providecommand{\urlprefix}{URL }
\providecommand{\doi}[1]{https://doi.org/#1}

\bibitem{blum1989}
Blum, L., Shub, M., Smale, S.: On a theory of computation over the real
  numbers; {NP} completeness, recursive functions and universal machines
  (extended abstract). [Proceedings 1988] 29th Annual Symposium on Foundations
  of Computer Science pp. 387--397 (1988)

\bibitem{BurgisserC06}
B{\"{u}}rgisser, P., Cucker, F.: Counting complexity classes for numeric
  computations {II:} algebraic and semialgebraic sets. J. Complex.
  \textbf{22}(2),  147--191 (2006)

\bibitem{DurandHKMV18}
Durand, A., Hannula, M., Kontinen, J., Meier, A., Virtema, J.: Approximation
  and dependence via multiteam semantics. Ann. Math. Artif. Intell.
  \textbf{83}(3-4),  297--320 (2018)

\bibitem{DHKMV18}
Durand, A., Hannula, M., Kontinen, J., Meier, A., Virtema, J.: Probabilistic
  team semantics. In: FoIKS. Lecture Notes in Computer Science, vol. 10833, pp.
  186--206. Springer (2018)

\bibitem{DurandKRV22}
Durand, A., Kontinen, J., de~Rugy{-}Altherre, N., V{\"{a}}{\"{a}}n{\"{a}}nen,
  J.: Tractability frontier of data complexity in team semantics. {ACM} Trans.
  Comput. Log.  \textbf{23}(1),  3:1--3:21 (2022)

\bibitem{galliani08}
Galliani, P.: {G}ame {V}alues and {E}quilibria for {U}ndetermined {S}entences
  of {D}ependence {L}ogic (2008), {MSc} Thesis. ILLC Publications,
  {MoL}--2008--08

\bibitem{GallianiH13}
Galliani, P., Hella, L.: Inclusion logic and fixed point logic. In: {CSL}.
  LIPIcs, vol.~23, pp. 281--295. Schloss Dagstuhl - Leibniz-Zentrum f{\"{u}}r
  Informatik (2013)

\bibitem{Gradel13}
Gr{\"{a}}del, E.: Model-checking games for logics of imperfect information.
  Theor. Comput. Sci.  \textbf{493},  2--14 (2013)

\bibitem{GradelG98}
Gr{\"{a}}del, E., Gurevich, Y.: Metafinite model theory. Inf. Comput.
  \textbf{140}(1),  26--81 (1998)

\bibitem{HannulaH22}
Hannula, M., Hella, L.: Complexity thresholds in inclusion logic. Inf. Comput.
  \textbf{287},  104759 (2022)

\bibitem{HannulaHKKV19}
Hannula, M., Hirvonen, {\AA}., Kontinen, J., Kulikov, V., Virtema, J.: Facets
  of distribution identities in probabilistic team semantics. In: {JELIA}.
  Lecture Notes in Computer Science, vol. 11468, pp. 304--320. Springer (2019)

\bibitem{HHK22}
Hannula, M., Hirvonen, M., Kontinen, J.: On elementary logics for quantitative
  dependencies. Ann. Pure Appl. Log.  \textbf{173}(10),  103104 (2022)

\bibitem{jelia23}
Hannula, M., Hirvonen, M., Kontinen, J., Mahmood, Y., Meier, A., Virtema, J.:
  Logics with probabilistic team semantics and the boolean negation. In: Gaggl,
  S., Martinez, M.V., Ortiz, M. (eds.) Logics in Artificial Intelligence. pp.
  665--680. Springer Nature Switzerland, Cham (2023)

\bibitem{HannulaKBV20}
Hannula, M., Kontinen, J., den Bussche, J.V., Virtema, J.: Descriptive
  complexity of real computation and probabilistic independence logic. In:
  {LICS}. pp. 550--563. {ACM} (2020)

\bibitem{DBLP:journals/corr/HannulaKLV16}
Hannula, M., Kontinen, J., L{\"{u}}ck, M., Virtema, J.: On quantified
  propositional logics and the exponential time hierarchy. In: GandALF.
  {EPTCS}, vol.~226, pp. 198--212 (2016)

\bibitem{HannulaKVV18}
Hannula, M., Kontinen, J., Virtema, J., Vollmer, H.: Complexity of
  propositional logics in team semantic. {ACM} Trans. Comput. Log.
  \textbf{19}(1),  2:1--2:14 (2018)

\bibitem{HannulaV22}
Hannula, M., Virtema, J.: Tractability frontiers in probabilistic team
  semantics and existential second-order logic over the reals. Ann. Pure Appl.
  Log.  \textbf{173}(10),  103108 (2022)

\bibitem{hodges97}
Hodges, W.: Compositional semantics for a language of imperfect information.
  Log. J. {IGPL}  \textbf{5}(4),  539--563 (1997)

\bibitem{Hyttinen15b}
Hyttinen, T., Paolini, G., V{\"{a}}{\"{a}}n{\"{a}}nen, J.: A logic for arguing
  about probabilities in measure teams. Arch. Math. Log.  \textbf{56}(5-6),
  475--489 (2017)

\bibitem{KontinenN11}
Kontinen, J., Nurmi, V.: Team logic and second-order logic. Fundam.
  Informaticae  \textbf{106}(2-4),  259--272 (2011)

\bibitem{kontinen_yang_2022}
Kontinen, J., Yang, F.: Complete logics for elementary team properties. The
  Journal of Symbolic Logic pp. 1--41 (2022). \doi{10.1017/jsl.2022.80}

\bibitem{Li22a}
Li, C.T.: First-order theory of probabilistic independence and single-letter
  characterizations of capacity regions. In: {ISIT}. pp. 1536--1541. {IEEE}
  (2022)

\bibitem{DBLP:phd/dnb/Luck20}
L{\"{u}}ck, M.: Team logic: axioms, expressiveness, complexity. Ph.D. thesis,
  University of Hanover, Hannover, Germany (2020),
  \url{https://www.repo.uni-hannover.de/handle/123456789/9430}

\bibitem{0072413}
Papadimitriou, C.H.: Computational complexity. Academic Internet Publ. (2007)

\bibitem{SchaeferS17}
Schaefer, M., Stefankovic, D.: Fixed points, nash equilibria, and the
  existential theory of the reals. Theory Comput. Syst.  \textbf{60}(2),
  172--193 (2017)

\bibitem{vaananen07}
V{\"{a}}{\"{a}}n{\"{a}}nen, J.A.: Dependence Logic - {A} New Approach to
  Independence Friendly Logic, London Mathematical Society student texts,
  vol.~70. Cambridge University Press (2007)

\end{thebibliography}

\end{document}